\documentclass[british]{scrartcl} \usepackage[T1]{fontenc}
\usepackage[latin9]{inputenc} 
\usepackage{float} 
\usepackage{bm}
\usepackage{amsmath} 
\usepackage{amssymb} 
\usepackage{graphicx}
\usepackage{esint} 
\usepackage[authoryear]{natbib} 
\usepackage{amsthm}
\usepackage{url}

\newcommand{\stn}{\|_\mathrm{E}} 

\makeatletter

\floatstyle{ruled} \newfloat{algorithm}{tbp}{loa}
\providecommand{\algorithmname}{Algorithm}
\floatname{algorithm}{\protect\algorithmname}

\usepackage{dsfont}
\usepackage{algpseudocode}

\newcommand{\R}{\mathbb{R}}

\newcommand{\setX}{\mathcal{X}} 

\newcommand{\ui}{[0,1)} 
\newcommand{\uid}{\ui^d}

 \newcommand{\E}{\mathbb{E}}
\newcommand{\var}{\mathrm{Var}} 
  
 \newcommand{\Q}{\mathbb{Q}} 
\newcommand{\Qh}{\widehat{\mathbb{Q}}} 
\newcommand{\Qb}{\overline{\mathbb{Q}}} 
\newcommand{\Qt}{\widetilde{\mathbb{Q}}} 

\newcommand{\Unif}{\mathcal{U}} 

\newcommand{\bx}{\mathbf{x}} 
\newcommand{\bu}{\mathbf{u}}
\newcommand{\bv}{\mathbf{v}} 
\newcommand{\by}{\mathbf{y}}
\newcommand{\bz}{\mathbf{z}}

\newcommand{\ind}{\mathds{1}} \newcommand{\dd}{\mathrm{d}}
\newcommand{\dx}{\dd \bx} \newcommand{\du}{\dd \mathbf{u}}
\newcommand{\dz}{\dd \mathbf{z}}  
\newcommand{\eqdef}{:=} 
\newcommand{\bigO}{\mathcal{O}} 
\newcommand{\smallo}{{\scriptscriptstyle\mathcal{O}}} 
\newcommand{\cvz}{\rightarrow 0} 
\renewcommand{\emptyset}{\varnothing} 
\newcommand{\opA}{\alpha_N}

\newcommand{\FK}{Feynman-Kac \,} \newcommand{\RN}{Radon-Nikodym\,}

\newcommand{\comment}[1]{ \ifthenelse{ \equal{\showcomment}{true} }{{\bf #1} }{} } 
\newcommand{\showcomment}{true}
\newcommand{\Sop}{\mathcal{S}} 

\newcommand{\HSFC}{H} 
\newcommand{\IHSFC}{h} 

\newtheorem{thm}{Theorem} 
\newtheorem{prop}{Proposition}
\newtheorem{corollary}{Corollary} 
\newtheorem{lemma}{Lemma}

\makeatother

\usepackage{babel} \date{}

\begin{document}
\author{Mathieu Gerber\thanks{Universit\'e de Lausanne, and CREST
    (Present address: Harvard University, Department of
    Statistics. Corresponding author, mathieugerber@fas.harvard.edu)}
  \and Nicolas Chopin\thanks{CREST-ENSAE}} \title{Convergence of
  sequential Quasi-Monte Carlo smoothing algorithms}

\maketitle

\begin{abstract}
  \citet{SQMC} recently introduced Sequential quasi-Monte Carlo (SQMC)
  algorithms as an efficient way to perform filtering in state-space
  models.  The basic idea is to replace random variables with
  low-discrepancy point sets, so as to obtain faster convergence
  than with standard particle filtering.  \citet{SQMC} describe briefly
  several ways to extend SQMC to smoothing, but do not provide supporting
  theory for this extension.  We discuss more
  thoroughly how smoothing may be performed within SQMC, and derive
  convergence results for the so-obtained smoothing algorithms. We
  consider in particular SQMC equivalents of  forward smoothing and
  forward filtering backward sampling, which are the most well-known smoothing techniques.
  As a preliminary step, we provide a generalization of the classical result of \citet{Hlawka1972} on the
  transformation of QMC point sets into low discrepancy point sets
  with respect to non uniform distributions. As a corollary of the
  latter, we note that we can slightly weaken the assumptions to prove
  the consistency of SQMC.


\end{abstract}

\section{Introduction}

State-space models are popular tools to model real life phenomena in
many fields such as Economics, Engineering and Neuroscience.  These
models are mainly used for extracting information about a hidden
Markov process $(\bx_t)_{t\geq0}$ of interest from a set of
observations $\by_{0:T}:=(\by_0,\dots,\by_T)$. In practice, this
typically translates to the estimation of $p(\bx_t|\by_{0:t})$, the
distribution of $\bx_t$ given the data $\by_{0:t}$, $0\leq t\leq T$
(called the \textit{filtering} distribution), and/or to
$p(\bx_{0:T}|\by_{0:T})$ (called the \textit{smoothing}
distribution). However, these distributions are intractable in most
cases, and require to be approximated in some way, the most popular
being particle filtering (Sequential Monte Carlo). See e.g. the books
of \cite{DouFreiGor}, \cite{CapMouRyd} for more background on
state-space models and particle filters.

Recently, \citet{SQMC} introduced sequential quasi-Monte Carlo (SQMC)
as an efficient alternative to particle filtering. Essentially, SQMC
amounts to replacing the random variates generated by a particle
filter with a QMC (low-discrepancy) point set; that is a set of $N$
points that are selected so as to cover more evenly the space that
random variates would; see e.g. the books of
\cite{Lemieux:MCandQMCSampling}, \cite{leobacherbook} for more
background on QMC.

\citet{SQMC} established that, for some constructions of RQMC
(randomised QMC) point sets, the convergence rate of SQMC (with
respect to $N$, the number of simulations) is at worst
$\bigO_P(N^{-1/2})$, while it is $\smallo_P(N^{-1/2})$ on the class of
continuous and bounded functions.  (This of course compares favourably
to the $\bigO_P(N^{-1/2})$ rate of particle filtering.)  In addition,
the numerical results of \citet{SQMC} show that SQMC dramatically
outperforms particle filtering in several applications.

One important question that remains however is how to use SQMC to
obtain smoothing estimates that converge as $N\rightarrow +\infty$.
Smoothing is recognised as a more difficult problem than
filtering \citep{briers2010smoothing}. 
Smoothing algorithms typically require extra steps on top
of particle filtering (such as a backward pass), and often cost
$\bigO(N^2)$ (but some variants cost $\bigO(N)$, as discussed later).

This paper discusses existing smoothing algorithms, explains
how they may be adapted to SQMC, and presents convergence
results for the corresponding SQMC smoothing algorithms.  We first study
forward smoothing, where trajectories are carried forward in
the particle filter, and show that this approach leads to consistent
estimates in SQMC.  Then, we derive a SQMC version of forward
filtering backward sampling (where complete trajectories are simulated
from the positions simulated by a particle filter, see
\citealp{GodsDoucWest}), and establish convergence results for the so
obtained smoothing estimates. We also consider the marginal version 
of backward sampling, which usually allows for a more precise estimation of 
marginal smoothing distributions.

The rest of this paper is organized as follows. Section
\ref{sec:preliminaries} introduces the model and the
notations considered in this work, and give a short description
of SQMC. Section \ref{sec:prelResults} contains some preliminary results that will be
needed to study SQMC smoothing. We first present a new consistency result for the 
forward step, which has the advantage to rely on weaker assumptions than in 
\citet{SQMC},  and state a result relative to the backward decomposition and SQMC 
estimation of the smoothing distribution. Then, we provide a generalization of the 
classical result of \citet{Hlawka1972} on the
transformation of QMC point sets into low discrepancy point sets
with respect to non uniform distributions that is essential to the
analysis of QMC smoothing algorithms. This section ends with some results on the 
conversion of discrepancies through the Hilbert space filling curve.  In Section \ref{sec:forw} we
establish the consistency of QMC forward smoothing while our results
on QMC forward-backward smoothing are given Section \ref{sec:back}. 
In Section \ref{sec:num} a numerical study examines the
performance of the QMC smoothing strategies discussed in this work
while Section \ref{sec:conclusion} concludes.

\section{Preliminaries}\label{sec:preliminaries}

\subsection{Model and related notations}\label{sec:model}


Let $(\bx_t)_{t\geq0}$ a Markov chain, defined on a space
$\setX\subseteq\mathbb{R}^d$ (equipped with Lebesgue measure), with initial distribution $m_0(\dx_0)$,
transition kernel $m_t(\bx_{t-1},\dx_t)$, $t\geq1$, and let
$(G_t)_{t\geq 0}$ a sequence of (measurable) potential functions,
$G_0:\setX\rightarrow\mathbb{R}^+$, $G_t:\setX\times \setX\rightarrow
\R^+$.
As in \citet{SQMC}, and most of the QMC
literature, we take $\setX=\ui^d$, but see Section 3 of 
\citet{SQMC} for how to generalise our results to unbounded state spaces. 

For this Feynman-Kac model $(m_t,G_t)_{t\geq 0}$, let 
$\Qb_t$ and $\Q_t$ be the  probability measures on $\setX$ such
that, for any  bounded measurable  function $\varphi:\setX\rightarrow
\mathbb{R}$, 
\begin{align*}
  \Qb_t(\varphi)& =\frac{1}{Z_{t-1}}\E\left[\varphi(\bx_t)G_0(\bx_0)\prod_{s=1}^{t-1}G_s(\bx_{s-1},\bx_s)\right]\\
  \Q_t(\varphi)& =\frac{1}{Z_{t}}\E\left[\varphi(\bx_t)G_0(\bx_0)\prod_{s=1}^{t}G_s(\bx_{s-1},\bx_s)\right]\\
  Z_{t}& =
  \E\left[G_{0}(\bx_{0})\prod_{s=1}^{t}G_{s}(\bx_{s-1},\bx_{s})\right]
\end{align*}
where expectations are with respect to the law of Markov chain $(\bx_t)$, 
and empty products equal one. Similarly, let 
 $\Qt_t$ be the probability measure on $\setX^{t+1}$
such that, for any bounded test function
$\varphi:\setX^{t+1}\rightarrow \mathbb{R}$,
$$
\Qt_t(\varphi)=\frac{1}{Z_{t}}\E\left[\varphi(\bx_{0:t})G_0(\bx_0)\prod_{s=1}^{t}G_s(\bx_{s-1},\bx_s)\right].
$$ 
In the sequel, the notation $0:t$ is used to denote the set of
integers $\{0,\dots,t\}$ and $\bx_{0:t}$ denotes the collection
$\{\bx_s\}_{s=0}^t$. Similarly, in
what follows we use the shorthand $\bx^{1:N}$ for a collection
$\{\bx^n\}_{n=1}^N$ of $N$ points in $\mathbb{R}^d$, and
$\bx_{0:t}^{1:N}$ for collection $\{\bx_{0:t}^n\}_{n=1}^N$ of $N$ points
in $\mathbb{R}^{(t+1)d}$. Finally, for a probability measure
$\pi\in\mathcal{P}(\setX)$, with $\mathcal{P}(\setX)$ the set of
probability measures on $\setX$ absolutely continuous with respect to
the Lebesgue measure, $\pi(\varphi)$ denotes the expectation of
$\varphi(\bx)$ under $\pi$.

To make more transparent the connection between this \FK formulation and state-space
models, assume Markov chain $(\bx_t)$ is observed indirectly
through $\by_t$, which has conditional probability density $f^Y(\by_t|\bx_t)$ (with 
respect to an appropriate measure, typically Lebesgue). 
If we take $G_0(\bx_0)=f^Y(\by_0|\bx_0)$, $G_t(\bx_{t-1},\bx_t)=f^Y(\by_t|\bx_t)$ 
for $t>0$, $\Q_t(\dx_t)$ becomes the filtering distribution (the law of $\bx_t|\by_{0:t}$),
$\Qb_t(\dx_t)$ the predictive distribution (the law
of $\bx_t|\by_{0:t-1}$), and $\Qt_t(\dx_{0:t})$, the object of
interest in this work, namely the smoothing distribution
(the law of $\bx_{0:t}|\by_{0:t}$). 
In addition, $Z_t$ is the marginal likelihood of observations $\by_{0:t}$.
In this case, $G_t$ depends only on $\bx_t$, but having a $G_t$ that
depends on both $\bx_{t-1}$ and $\bx_t$ makes it possible to apply our results
to a more general class of algorithms (such as those 
where the Markov transition used to simulate particles differs from 
the Markov transition of the model). 

\subsection{Extreme norm and QMC  point sets\label{sub:QMC}}

As in \citet{SQMC}, our consistency results
are stated in term of the \textit{extreme} norm, defined, for two
probability measures $\pi_1$ and $\pi_2$ on $\ui^d$, by
$$
\|\pi_1-\pi_2\stn=\sup_{B\in\mathcal{B}_{\ui^d}}\left|\pi_1(B)-\pi_2(B)\right|
$$ 
where
$$\mathcal{B}_{\ui^d}=\{B=\prod_{i=1}^d[a_i,b_i],\, 0\leq
a_i<b_i<1\}.
$$  
Note that $\|\pi_N-\pi\stn\rightarrow 0$ implies that $\pi_N(\varphi)\rightarrow \pi(\varphi)$ for any bounded and continuous function $\varphi$ (by portmanteau lemma,
see e.g. Lemma 2.2, p.6 of \citealt{VanderVaart2007}). In words, consistency 
for the extreme norm implies consistency of estimates for test functions $\varphi$
that are bounded and continuous.

The extreme norm is natural in QMC contexts since it
can be viewed as the generalization of the \textit{extreme discrepancy}
 of a point set $\bu^{1:N}$ in $\ui^d$, defined by
$$ D(\bu^{1:N})=\|\Sop(\bu^{1:N})-\lambda_d\stn
$$ 
where $\lambda_d$ denotes the Lebesgue measure on $\mathbb{R}^d$
and $\Sop$ is the operator 
$$\Sop:\bu^{1:N} \rightarrow \frac 1 N \sum_{n=1}^N \delta_{\bu^n}.
$$
The extreme discrepancy therefore
measures how a point set spreads evenly over $\ui^d$ and is used to
define formally QMC point sets. To be more specific, $\bu^{1:N}$ is a
QMC point set in $\ui^d$ if $D(\bu^{1:N})=\bigO(N^{-1}(\log
N)^{d})$. Note that, for a sample $\bu^{1:N}$ of $N$ IID uniform
random numbers in $\ui^d$, $D(\bu^{1:N})=\bigO(N^{-1/2}\log\log N)$
almost surely by the law of iterated logarithm \citep[see e.g.][page
167]{Niederreiter1992}.  There exist many constructions of QMC point
sets in the literature \citep[see][for more details on this
topic]{Niederreiter1992,dick2010digital} and, although we write
$\bu^{1:N}$ rather than $\bu^{N,1:N}$, $\bu^{1:N}$ may not necessarily
be the $N$ first points of a fixed sequence, i.e. one may have
$\bu^{N,N-1}\neq\bu^{N-1,N-1}$. However, it is worth keeping in mind
that all the results presented in this paper hold both for point sets and
sequences.

Even if in this work we are mainly interested in consistency results (which hold
for deterministic point sets $\bu^{1:N}$),
we will sometimes refer to randomized QMC (RQMC) point
sets. Formally, $\bu^{1:N}$ is RQMC point set if it is a QMC point set
with probability one and if, marginally, $\bu^n\sim\Unif(\ui^d)$ for
all $n\in 1:N$.

\subsection{The Hilbert space-filling curve}

The Hilbert space filling curve plays a key role in the construction and the
analysis of SQMC. This curve is a 
H\"{o}lder continuous fractal map  $H:[0,1]\rightarrow [0,1]^d$ 
 that fills completely $[0,1]^d$; see Figure \ref{fig:hilbert} for a graphical depiction, and 
Appendix \ref{app:Hilbert} for a presentation of its  mains properties. 
In what follows, we denote by $h:[0,1]^d\rightarrow [0,1]$
its pseudo-inverse which verifies, for any $\bx\in [0,1]^d$, $H\circ
h(\bx)=\bx$, and, 
for $d=1$, we use the natural convention that $H$
and $h$ are the identity mappings, i.e. $H(x)=h(x)=x$, $\forall x\in
[0,1]$.  
The Hilbert curve is not uniquely defined;
in this work, we assume that $H$ is such that
$H(0)=\mathbf{0}\in[0,1]^d$ \citep[this is in fact the classical way
to construct the Hilbert curve, see e.g.][]{Hamilton2008}. This
technical assumption is needed in order to be consistent with the fact
that we work with left-closed and right-opened hypercubes since, in
that case, $h(\ui^d)=\ui$. 

\begin{figure}
\begin{centering}
\includegraphics[scale=0.4]{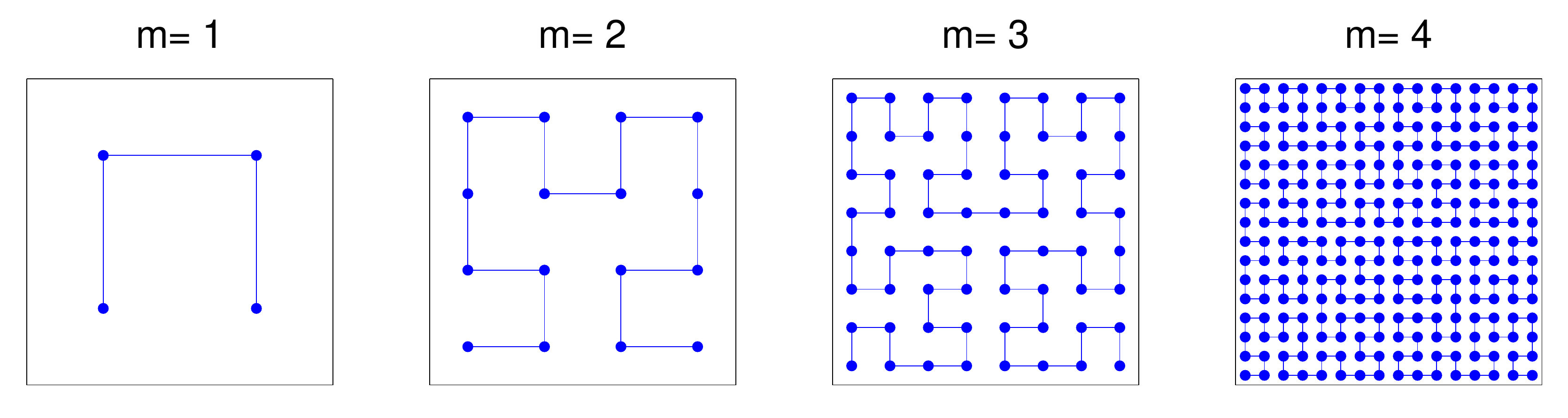}
\caption{First four iterates of sequence $H_m$, the limit of which is 
the Hilbert curve $H$, for $d=2$ (source: \cite{HeOwen2014})\label{fig:hilbert}}
\end{centering}
\end{figure}

\subsection{Rosenblatt transform}

Another important technical tool for SQMC is the Rosenblatt transform. 
For a probability distribution $\pi$ over $\ui$, $F_\pi$ denotes its CDF (cumulative distribution),
and $F_\pi^{-1}$ its inverse CDF; i.e. $F_\pi^{-1}=\inf\{x\in\ui:\,F(x)\geq u\}$. More generally,
for  a probability distribution $\pi$ over $\setX=\uid$, $F_\pi$  denotes the Rosenblatt
transform, that is
\[
F_{\pi}(\bx)=\left(F_{\pi,1}(x_{1}),F_{\pi,2}(x_{2}|x_{1}),\ldots,F_{\pi,d}(x_{d}|x_{1:d-1})\right)^T,\quad\bx=(x_{1},\ldots,x_{d})^T\in\setX,
\]
where $F_{\pi,1}$ is the CDF of the marginal distribution of the first
component (relative to $\pi$), and for $i\geq2$,
$F_{\pi,i}(\cdot|x_{1:i-1})$ is the CDF of component $x_{i}$,
conditional on $(x_{1},\ldots,x_{i-1}$), again relative to $\pi$.
The inverse of $F_\pi$ is denoted $F_\pi^{-1}$. 
Note how the Rosenblatt transform and its inverse define a monotonous map
that transforms any distribution into a uniform distribution, and vice-versa. 

We overload this notation for Markov kernels: $F_{m_t}(\bx_{t-1},\cdot)$ is the 
the Rosenblatt transform of probability distribution $m_t(\bx_{t-1},\dx_t)$ (for
fixed $\bx_{t-1}\in\setX)$, and $F_{m_t}^{-1}$ is defined similarly.

\subsection{Sequential quasi-Monte Carlo}\label{subsec:sqmc}

The basic structure of SMC (Sequential Monte Carlo, also known as particle filtering) 
algorithms is recalled as Algorithm \ref{alg:SMC}.
One sees from this description that SMC is a class of iterative algorithms 
that use resampling and mutation steps to move from a
discrete approximation $\Qh_{t}^N(\dx_{t})$ of $\Q_{t}(\dx_t)$ to a discrete 
approximation $\Qh_{t+1}^N(\dx_{t+1})$ of $\Q_{t+1}(\dx_{t+1})$, where
$$
\Qh_t^N(\dx_t) = \sum_{n=1}^N W_t^n \delta_{\bx_t^n}(\dx_t),\quad t\in 0:T.
$$

\begin{algorithm}
  \caption{SMC Algorithm\label{alg:SMC}}
  \begin{algorithmic}
    \State Generate (for $n\in 1:N$)  $\bx_0^n\sim m_0(\dx_0)$
    \State Compute (for $n\in 1:N$) 
    $W_{0}^{n}=G_{0}(\bx_{0}^{n})/\sum_{m=1}^{N}G_{0}(\bx_{0}^{m})$ 
    \For{$t=1$ to $t=T$} 
    \State Generate (for $n\in 1:N$)   $u_t^n\sim \mathcal{U}\ui$ and
set $a_t^n=F_{t,N}^{-1}(u_t^n)$, where $F_{t,N}(m)=\sum_{n=1}^N W_t^n\ind(n\leq m)$
	 
    \State Generate (for $n\in 1:N$) $\bx^n_t \sim m_t(\hat{\bx}^n_{t-1},\dx_t)$, where $\hat{\bx}_{t-1}^n=\bx_{t-1}^{a_t^n}$
    \State Compute (for $n\in 1:N$) 
    $W_{t}^{n}=G_{t}(\hat{\bx}_{t-1}^{n},\bx_{t}^{n})/\sum_{m=1}^{N}G_{t}(\hat{\bx}_{t-1}^{m},\bx_{t}^{m})$
    \EndFor
  \end{algorithmic}
\end{algorithm}
A closer look at Algorithm \ref{alg:SMC} shows that, for $t\geq 1$, the resampling and the mutation steps together amounts to sampling from the (random) distribution on $\setX^2$ defined by
\begin{align}\label{eq:IS}
\pi^N_t(\dd(\bx_{t-1},\bx_t))=\Qh^N_{t-1}\otimes m_t(\dd(\bx_{t-1},\bx_t))
\end{align}
where, for a
probability measure $\pi\in\mathcal{P}(\ui^{d_1})$ and a kernel
$K:\ui^{d_1}\rightarrow\mathcal{P}(\ui^{d_2})$,  the notation $\pi\otimes
K(\dd(\bx_1,\bx_2))$ denotes the probability measure
$\pi(\dx_1)K(\bx_1,\dx_2)$ on $\ui^{d_1+d_2}$.

Based on this observation, the basic idea of SQMC is to replace the
uniform random numbers used at iteration $t\geq 1$ of an SMC algorithm to sample from \eqref{eq:IS} by a QMC point set $\bu_t^{1:N}$ of appropriate
dimension. In the deterministic version of SQMC, the only known
property of $\bu_t^{1:N}$ is that its discrepancy converges to zero
as $N$ goes to infinity. Thus, we must make sure that 
the transform applied to $\bu_t^{1:N}$ preserves consistency 
(relative the extreme norm): i.e. 
$D(\bu^{1:N})\cvz$ implies that 
$\|\Gamma^N_t(\bu^{1:N})-\pi^N_t\stn\cvz$, where $\Gamma^N_t$ is the
chosen transformation.

 When the state-space  is univariate, \citet{SQMC} propose to use for $\Gamma_t^N$ the inverse  Rosenblatt
transformation of $\pi_t^N$ described in the previous subsection, which  amounts  to sample $(\hat{x}_{t-1},x_t)$ from \eqref{eq:IS} as follows:
$$
\hat{x}_{t-1}=F_{\Qh_{t-1}^N}^{-1}(u_t),\quad x_{t}=F^{-1}_{m_t}(\hat{x}_{t-1}, v_t),\quad (u_t,v_t)\sim\Unif(\ui^2).
$$
However, when the state variable is multivariate (i.e. $d>1$) this approach cannot be directly  used because in that case $\Qh_{t-1}^N(\dx_{t-1})$ is a weighted sum of Dirac measures over $\setX\subset\mathbb{R}^d$.

To extend this approach to multidimensional state-space models, \citet{SQMC} transform the multivariate distribution $\Qh_{t-1}^N(\dx_{t-1})$ into a univariate  distribution $\Qh_{t-1,h}^N(\dd h_{t-1})$ by applying the change of variable $h:\bx\in\setX\rightarrow \ui$, where $h$ is the pseudo-inverse of the Hilbert curve (see Section \ref{sec:hilbert}). Using this change of variable, the resampling and mutation steps of SMC are equivalent to sampling from
\begin{align}\label{eq:ISh}
\pi^N_{t,h}(\dd(h_{t-1},\bx_t))=\Qh^N_{t-1,h}\otimes m_{t,h}(\dd(h_{t-1},\bx_t))
\end{align}
where $m_{t,h}(h_{t-1},\bx_t):=m_{t}(H(h_{t-1}),\bx_t)$. As for the univariate setting,  one can generate random variates form $\pi^N_{t,h}(\dd(h_{t-1},\bx_t))$ using the inverse Rosenblatt transformation of this distribution; that is,  we can sample $(\hat{h}_{t-1},\bx_t)$ from \eqref{eq:ISh} as follows:
$$
\hat{h}_{t-1}=F_{\Qh_{t-1,h}^N}^{-1}(u_t),\quad \bx_{t}=F^{-1}_{m_{t}}(H(\hat{h}_{t-1}), \bv_t),\quad (u_t,\bv_t)\sim\Unif(\ui^{d+1}).
$$
The resulting SQMC algorithm, which is therefore based for $t\geq 1$ on $d+1$-dimensional QMC point sets $\bu_t^{1:N}$, $\bu_t^n=(u^n_t,\bv^n_t)\in\ui^{d+1}$,  is presented in Algorithm \ref{alg:SQMC_B}.

\begin{algorithm}
  \caption{SQMC Algorithm\label{alg:SQMC_B}}
  \begin{algorithmic}
    \State Generate a QMC point set $\bu_{0}^{1:N}$
    in $\ui^{d}$ 
    \State Compute (for $n\in 1:N$)  $\bx_{0}^{n}=F_{m_0}^{-1}(\bu_{0}^{n})$
    \State Compute (for $n\in 1:N$) 
    $W_{0}^{n}=G_{0}(\bx_{0}^{n})/\sum_{m=1}^{N}G_{0}(\bx_{0}^{m})$ 
    \For{$t=1$ to $t=T$} 
    \State Generate a QMC point set $\bu_{t}^{1:N}$ in $\ui^{d+1}$, let
    $\bu_{t}^{n}=(u_{t}^{n},\bv^n_t)$, where $u_t^n\in\ui$, $\bv_t^n\in\uid$. Assume that, for all $n,m\in 1:N$, $n\leq m\implies u_t^n\leq u_t^m$
    \State\label{step1} Hilbert sort: find permutation $\sigma_{t-1}$
    such that $\IHSFC(\bx_{t-1}^{\sigma_{t-1}(1)})
    \leq\ldots\leq\IHSFC(\bx_{t-1}^{\sigma_{t-1}(N)})$ 
    \State Compute (for $n\in 1:N$) 
    $a_{t-1}^n=F_{t,N}^{-1}(u_t^{n})$ where
    $F_{t,N}(m)=\sum_{i=1}^NW^{\sigma_{t-1}(i)}_{t-1}\mathbb{I}(i\leq
    m)$ 
    \State \label{step2} Compute (for $n\in 1:N$) 
    $\bx_{t}^{n}=F_{m_t}^{-1}(\hat{\bx}_{t-1}^n,\bv_{t}^{n})$,
    where $\hat{\bx}_{t-1}^n=\bx_{t-1}^{a_{t-1}^n}$ 
    \State Compute (for $n\in 1:N$) 
    $W_{t}^{n}=G_{t}(\hat{\bx}_{t-1}^{n},\bx_{t}^{n})/\sum_{m=1}^{N}G_{t}(\hat{\bx}_{t-1}^{m},\bx_{t}^{m})$
    \EndFor
  \end{algorithmic}
\end{algorithm}

\section{Preliminary results\label{sec:prelResults}}

\subsection{Consistency of SQMC}

The consistency of Algorithm \ref{alg:SQMC_B} (as $N\rightarrow +\infty$, 
with respect to the extreme metric) 
was established in \citet[][Theorem 5]{SQMC}, under the assumption that
$F_{m_t}$ is Lipschitz. We generalise below this result to the case 
where $F_{m_t}$ is H\"{o}lder continuous, as this generalisation 
will be needed later on when dealing with the backward step. 
This also allows us to recall some of the assumptions that will be repeated
throughout the paper. For convenience, let $F_{m_t}(\bx_{t-1},\bx_t)=F_{m_0}(\bx_0)$ when $t=0$.

\begin{thm}
  \label{thm:consistency2} Consider the set-up of Algorithm
  \ref{alg:SQMC_B} where, for all $t\in0:T$, $(\bu_{t}^{1:N})_{N\geq
    1}$ is a sequence of point sets in $\ui^{d_{t}}$, with $d_{0}=d$
  and $d_{t}=d+1$ for $t>0$, such that $D(\bu_t^{1:N})\rightarrow 0$
  as $N\rightarrow+\infty$.  Assume the following holds for all $t\in
  0:T$:
  \begin{enumerate}
  \item \label{H:thmPF1_B:1} The $\bx_t^n$'s are
    pairwise distinct: $\bx_t^n\neq \bx_t^m$ for $n\neq m \in 1:N$;
  \item \label{H:thmPF1_B:2} $G_{t}$ is continuous and bounded;
  \item \label{H:thmPF1_B:3} $F_{m_{t}}(\bx_{t-1},\bx_{t})$ is such
    that
    \begin{enumerate}
    \item\label{H:thmPF1_B:3a} For $i\in 1:d$ and for a fixed $\bx'$,
      the $i$-th coordinate of $F_{m_t}\left(\bx',\bx\right)$ is
      strictly increasing in $x_{i}\in\ui$, the $i$-th coordinate of $\bx$;
    \item\label{H:thmPF1_B:3b} Viewed as a function of $\bx'$ and
      $\bx$, $F_{m_t}\left(\bx',\bx\right)$ is H\"{o}lder continuous;
    \item\label{H:thmPF1_B:3c} For $i\in 1:d$,
      $m_{ti}(\bx',x_{1:i-1},\dd x_i)$, the distribution of the
      component $x_i$ conditional on $(x_1,..,x_{i-1})$ relative to
      $m_{t}(\bx',\dx)$, admits a density
      $p_{ti}(x_i|x_{1:{i-1}},\bx')$ with respect to the Lebesgue
      measure such that $\|p_{ti}(\cdot|\cdot)\|_{\infty}<+\infty$.
    \end{enumerate}
  \item \label{H:thmPF1_B:4}
    $\Q_{t}(\dx_{t})=p_{t}(\bx_{t})\lambda_d(\dx_{t})$ where
    $p_{t}(\bx_{t})$ is a strictly positive bounded density.
  \end{enumerate}
 For $t\in 1:T$, let $P^N_{t,h}=(h(\hat{\bx}_{t-1}^{1:N}),\bx_{t}^{1:N})$. Then, under
  Assumptions \ref{H:thmPF1_B:1}-\ref{H:thmPF1_B:4}, we have, for
  $t\in 1:T$,
  \[
  \left\Vert \Sop(P_{t,h}^N)-\Q_{t-1,h}\otimes m_{t,h}
  \right\stn\cvz,\quad\mbox{as }N\rightarrow+\infty
  \]
  and, for $t\in 0:T$,
  \[
  \|\Qh_{t}^{N}-\Q_{t}\stn\cvz,\quad\mbox{as }N\rightarrow+\infty.
  \]
\end{thm}

The difference with \citet[][Theorem 5]{SQMC} is Assumption
\ref{H:thmPF1_B:3}, where \ref{H:thmPF1_B:3c} was not needed but it
was assumed that $F_{m_t}$ is a Lipschitz function. In this work,
Assumption \ref{H:thmPF1_B:3c} will be required to establish the
validity of the backward step.  Assumption \ref{H:thmPF1_B:1} is a
technical condition that is verified almost surely for the randomized
version of SQMC while assuming that $G_t$ is bounded is standard in
particle filtering \citep{DelMoral:book}.
In our notations, we drop the dependence of point sets on $N$, i.e. we write $(\bx^{1:N})_{N\geq 1}$
rather than $(\bx^{N,1:N})_{N\geq 1}$, although in full generality
$\bx^{1:N}$ may not necessarily be the $N$ first points of a fixed
sequence.

The proof of Theorem \ref{thm:consistency2} is omitted since it can be
directly deduced from the proof of \citet[][Theorem 5]{SQMC} and from
the generalization of the result of \citet[][``Satz 2'']{Hlawka1972}
presented in the Section \ref{subsec:Hlawka}, which constitutes one of the key
ingredients to study the backward pass of SQMC.

\subsection{Backward decomposition}

Backward smoothing algorithms require that Markov kernel $m_t(\bx_{t-1},\dx_t)$ admits a (strictly positive) 
probability density which may be computed pointwise; 
$m_t(\bx_{t-1},\dx_t)=m_t(\bx_{t-1},\bx_t)\lambda_d(\dx_t)$, with 
$m_t(\bx_{t-1},\dx_t)>0$ (and $\lambda_d$ being Lebesgue measure in our case).

The backward decomposition of the smoothing
distribution is \citep[e.g.][]{DelMoral2010}:
\begin{align}\label{backward:eq:backS}
  \Qt_{T}(\dx_{0:T})=\Q_T(\dx_T)
  \prod_{t=1}^T\mathcal{M}_{t,\Q_{t-1}}(\bx_{t},\dx_{t-1})
\end{align}
where, for any $\pi\in\mathcal{P}(\setX)$ and $t\in 1:T$,
$\mathcal{M}_{t,\pi}:\setX\rightarrow \mathcal{P}(\setX)$ is the Markov
kernel such that
$$
\mathcal{M}_{t,\pi}(\bx_{t},\dx_{t-1})\eqdef
\tilde{G}_t(\bx_{t-1},\bx_t)
\pi(\dx_{t-1})
$$
with 
\begin{equation}\label{eq:Gtilde}
\tilde{G}_t(\bx_{t-1},\bx_t)\eqdef
\frac{m_{t}(\bx_{t-1},\bx_{t})G_t(\bx_{t-1},\bx_t)}
{\int_{\setX} m_t(\tilde\bx_{t-1},\bx_t) G_t(\tilde\bx_{t-1},\bx_t) \pi(\dd\tilde\bx_{t-1})}.
\end{equation}
As a preliminary result, we show that the plug-in estimate
$\Qt_{T}^{N}$ of $\Qt_T$, obtained by replacing $\Q_t$ with 
 $\Qh_t^N$ in \eqref{backward:eq:backS}, is consistent for the
extreme norm; see Appendix \ref{p-thm:BS} for a proof.

\begin{thm}\label{thm:BS}
Consider the set-up of Algorithm \ref{alg:SQMC_B}, define for $ t\in 1:T$
\begin{equation} \label{eq:backwardN}
\Qt_{t}^{N}(\dx_{0:t})=\Qh_t^N(\dx_t)
\prod_{s=1}^t\mathcal{M}_{s,\Qh^N_{s-1}}(\bx_s,\dx_{s-1}),
\end{equation} 
and consider the following hypotheses: 
\begin{description}
\item[] (H1) $\tilde{G}_t$  is continuous and bounded, $\|\tilde{G}_t\|_\infty<\infty$;
\item[] (H2)  $F_{\mathcal{M}_{t,\Q_{t-1}}}(\bx_t,\bx_{t-1})$ satisfies
Assumptions 3a and 3b of Theorem \ref{thm:consistency2} (i.e. replace $m_t$ by $\mathcal{M}_{t,\Q_{t-1}}$ in these assumptions).
\end{description}

Then,
\begin{enumerate}
\item Under (H1) and the assumptions of Theorem \ref{thm:consistency2}, one has
(for $ t\in 1:T$)
  \begin{align}\label{eq:Smooth}
    \sup_{\bx_t\in\ui^d}\|\mathcal{M}_{t,\Qh^N_{t-1}}(\bx_t,\dx_{t-1})-\mathcal{M}_{t,\Q_{t-1}}(\bx_t,\dx_{t-1})\stn\cvz,\quad\mbox{as
    }N\rightarrow+\infty.
  \end{align}
\item If \eqref{eq:Smooth} holds, and under  (H2) and the assumptions of Theorem \ref{thm:consistency2}, one has
(for $t\in 1:T$)
  \begin{align}\label{eq:Smooth2}
    \|\Qt_t^{N}-\Qt_t\stn\cvz,\quad\mbox{as }N\rightarrow+\infty.
  \end{align}
\end{enumerate}

\end{thm}


The first result above does not have a clear interpretation, but it will
be used as an intermediate result later on. 


\subsection{A generalization of Satz 2 of \citet{Hlawka1972}\label{subsec:Hlawka}}

Theorem \ref{thm:GenHM} below generalizes Proposition `Satz 2' of \cite{Hlawka1972}
to the case where point sets in $\ui^d$ are transformed through a
H\"older continuous Rosenblatt transformation; see Appendix
\ref{p-thm:GenHM} for a proof.

\begin{thm}\label{thm:GenHM} Let $\pi$ be a probability measure on
  $\ui^d$ and assume the following:
  \begin{enumerate}
  \item \label{H:thmHM:1} Viewed as a function of $\bx$,
    $F_{\pi}\left(\bx\right)$ is H\"{o}lder continuous with H\"{o}lder
    exponent $\kappa\in (0,1]$;
  \item \label{H:thmHM:2} For $i\in 1:d$, the $i$-th coordinate of
    $F_{\pi}\left(\bx\right)$ is strictly increasing in $x_{i}\in\ui$, the $i$-th coordinate of
$\bx$;

  \item \label{H:thmHM:3} For $i\in 1:d$, $\pi_i(x_{1:i-1},\dd x_i)$,
    the distribution of the component $x_i$ conditional on
    $(x_1,..,x_{i-1})$ relative to $\pi(\dx)$, admits a density
    $p_i(x_i|x_{1:{i-1}})$ with respect to the Lebesgue measure such
    that $\|p_i(\cdot|\cdot)\|_{\infty}<+\infty$.
  \end{enumerate}
  Let $\bu^{1:N}$ be a point set in $\ui^d$ and, for $n\in 1:N$,
  define $\bx^n=F_{\pi}^{-1}(\bu^n)$. Then, for a constant $c>0$,
  \[
  \|\Sop(\bx^{1:N})-\pi\stn\leq c D(\bu^{1:N})^{1/\tilde{d}}
  \]
  where $\tilde{d}=\sum_{i=0}^{d-1}\lceil \kappa^{-1} \rceil^i$.
\end{thm}

When the Rosenblatt transformation $F_{\pi}$ is Lipschitz,
$\tilde{d}=d$ and we recover the result of \citet{Hlawka1972}. In this
case, Assumption \ref{H:thmHM:3} is not needed. Notice that the rate
provided in Theorem \ref{thm:GenHM} decreases quickly with the
H\"older exponent $\kappa$. For $\kappa=1/2$, the convergence rate is
of order $\bigO(D(\bu^{1:N})^{1/2^d-1})$ and hence is very slow even
for moderate values of $d$.



We will see that the backward step of the forward-backward SQMC
smoothing algorithm amounts to applying to QMC point sets
transformations that are ``nearly'' $(1/d)$-H\"older continuous (in a
sense that we will make precise). The main message of Theorem
\ref{thm:GenHM}, as far as SQMC is concerned, is that such an
algorithm may be consistent (as $N\rightarrow +\infty$) despite being
based on non-Lipschitz transformations.

Theorem \ref{thm:GenHM} is interesting more generally, 
since the construction of low discrepancy point sets with respect to
non uniform distributions is an important problem, which is motivated
by the generalized Koksma-Hlawka inequality \citep[][Theorem
1]{Aistleitner2014}:
$$
\left|\frac{1}{N}\sum_{n=1}^N\varphi(\bx^n)-\int_{\ui}\varphi(\bx)\pi(\dx)\right|\leq
V(\varphi)\|\Sop(\bx^{1:N})-\pi\stn
$$ 
where $V(\varphi)$ is the variation of $\varphi$ in the sense of
Hardy and Krause. It is also interesting to mention that the inverse
Rosenblatt transformation is the best known construction of low
discrepancy point sets for non uniform probability measures, although
the bounds for the extreme metric given in \citet[][``Satz
2'']{Hlawka1972} and in Theorem \ref{thm:GenHM} are very far from the
best known achievable rate since \citet[][Theorem 1]{Aistleitner2013} have
established the existence, for any probability measure $\pi$ on $\ui^d$,
of a sequence $(\bx^{n})_{n\geq 1}$ verifying
$\|\Sop(\bx^{1:N})-\pi\stn=\bigO(N^{-1}(\log N)^{0.5(3d+1)})$.

\subsection{Discrepancy conversion through the Hilbert space filling curve}
\label{sec:hilbert}

We now state results regarding how the Hilbert curve 
$H:[0,1]\rightarrow [0,1]^d$ conserves discrepancy. 
Such results were not directly needed to establish the consistency 
of SQMC. Indeed, as outlined in  the statement of Theorem
\ref{thm:consistency2}, it was sufficient to show that 
$P_{t,h}^N$ has low discrepancy with respect to the proposal distribution
$\Q_{t-1,h}\otimes m_{t,h}$, where  we recall that
$P_{t,h}^N=\left(h(\hat{\bx}_{t-1}^{1:N}),\bx_t^{1:N}\right)$, with
$h(\hat{\bx}_{t-1}^{1:N})\in\ui$. The discrepancy of the 
``resampled'' particles $\hat{\bx}_{t-1}^{1:N}$ in $\uid$ was not derived. But, again, we will need such results when dealing with backward estimates. 

More precisely, and as explained below (see Section \ref{subsec:BS}), the analysis of these latter require results on the conversion of discrepancies through the following mapping, defined for $k\in\mathbb{N}$, by
\begin{align}\label{eq:HT}
  H_k:(x_0,\dots,x_k)\in\ui^{(k+1)}\mapsto
  (H(x_0),\dots,H(x_k))\in\ui^{d(k+1)}
\end{align}
and with  pseudo-inverse $h_k:\ui^{d(k+1)}\rightarrow\ui^{k+1}$.

Theorem \ref{thm:Hilbert2} and Corollary \ref{cor:Hilbert2} below  are  generalizations of \citet[][Theorem 1]{Schretter2015}, which corresponds to Theorem  \ref{thm:Hilbert2} with $k=0$, $\pi_h$ the uniform distribution on $\ui$ and $\pi_h^N=\Sop(u^{1:N})$ for a point set $u^{1:N}$ in $\ui$. To save space, the proofs the these two results are omitted.


\begin{thm}\label{thm:Hilbert2}
  Let $\pi(\dx)=\pi(\bx)\lambda_{d(k+1)}(\dx)$, $k\in\mathbb{N}$, be a probability measure on $\ui^{d(k+1)}$ with
  bounded density $\pi$, $\pi_{h_k}$ be the image of $\pi$ by $h_k$ and
  $(\pi_{h_k}^N)_{N\geq 1}$ be a sequence of probability measures on $\ui^{k+1}$
  such that $\|\pi_{h_k}^N-\pi_{h_k}\stn\cvz$ as $N\rightarrow+\infty$.
  Let $\pi^N$ be the image by $H_k$ of  $\pi_{h_k}^N$.
  Then, 
$$ \|\pi^N-\pi\stn\cvz,\quad\mbox{as
}N\rightarrow+\infty. 
$$ 
\end{thm}

\begin{corollary}
  \label{cor:Hilbert2} Consider the set-up of Theorem \ref{thm:Hilbert2} with $k=0$ and let
  $K:\ui^{d}\rightarrow\mathcal{P}\left(\ui^{s}\right)$ be a Markov
  kernel, $K_h(h_1,\dx_2)=K(h(\bx_1),\dx_2)$ and
  $P_h^{N}=(h_1^{1:N},\bx_2^{1:N})$ be a sequence of point sets in
  $\ui^{1+s}$ such that, as $N\rightarrow +\infty$,
  $\|\Sop(P_h^{N})-\pi_h\otimes K_h\stn\rightarrow0$. Let $P^{N}=
  \left(H(h_1^{1:N}),\bx_2^{1:N}\right) $. Then,
  \[
  \|\Sop(P^{N})-\pi^N\otimes K\stn \cvz,\quad\mbox{as
}N\rightarrow+\infty.
  \]
\end{corollary}

A direct consequence of this corollary is that, under the assumptions of Theorem
\ref{thm:consistency2}, the point set
$P_t^N=(\hat{\bx}^{1:N}_{t-1},\bx_t^{1:N})$ is such that, as
$N\rightarrow+\infty$,
$$ \|\Sop(P_t^N)-\Q_{t-1}\otimes m_{t}\stn\cvz.
$$ 

Another consequence of this corollary is that Algorithm
\ref{alg:SQMC_B} can be trivially adapted to forward smoothing, 
as briefly explained in the next section.

\section{SQMC forward smoothing}\label{sec:forw}

Consider now the following extension of Algorithm \ref{alg:SQMC_B}, 
where full trajectories $\bz_t\eqdef\bx_{0:t}\in\setX^{t+1} $ are carried forward: at time $0$, set $\bz_0^n\eqdef\bx_0^n$, and, recursively, 
$\bz_t^n\eqdef(\hat{\bz}_t^n,\bx_t^n)$, with $\hat{\bz}_t^n\eqdef\bz_{t-1}^{a_{t-1}^n}$. 
In addition, replace the Hilbert sort step of Algorithm \ref{alg:SQMC_B}
by the same operation on full trajectories: 
\begin{quote}
  Hilbert sort: find permutation $\sigma_{t-1}$ such that
  $h^{t}(\bz_{t-1}^{\sigma_{t-1}(1)}) \leq\ldots\leq
  h^{t}(\bz_{t-1}^{\sigma_{t-1}(N)})$
\end{quote}
with $h^{t}$ the inverse of a Hilbert curve $H^t$ that maps $[0,1]$
into $[0,1]^{dt}$. In other words, this is the SQMC equivalent of the smoothing
technique known as `forward smoothing'.

\begin{prop}\label{prop:constforward}
Under Assumptions 1-3 of Theorem \ref{thm:consistency2}, and Assumption 4'
\begin{quote}
  \emph{4'. $\Qt_{t}(\dz_{t})=\tilde{p}_{t}(\bz_{t})\lambda_{d(t+1)}(\dz_{t})$
    where $\tilde{p}_t(\bz_{t})$ is a strictly positive bounded
    density;}
\end{quote}
one has, for $t\geq 0$ and the forward smoothing algorithm described above,
\begin{align}\label{eq:forward}
  \left\|\sum_{n=1}^NW^{n}_{t}\delta_{\bz_{t}^{n}}-\Qt_{t}\right\stn\cvz,\quad\mbox{as
  }N\rightarrow+\infty
\end{align}
where $\Qt_{t}$ denotes the smoothing distribution at time $t$.
\end{prop}

See Appendix \ref{p-thm:consforward} for a proof. 

This result is presented for the sake of completeness, 
but it is clear that it is of limited practical interest. 
Transformations through $H^t$ will lead to poor convergence rates as soon
as $t$ becomes large, as per Theorem \ref{thm:Hilbert2}. 
In addition, there is no reason to believe that the SQMC version of 
forward smoothing would not suffer from the same major drawback as its Monte Carlo
counterpart, namely that the $N$ simulated paths quickly coalesce to 
a single ancestor. 

\section{SQMC backward smoothing}\label{sec:back}

We now turn to the derivation and analysis of a SQMC version of backward smoothing. There exist in fact two backward smoothing algorithms. The first one 
\citep{DouGodAnd}, approximates the marginal smoothing distributions $\Q_{t|T}(\dx_t)$  for $t\in 0:T$; that is, the marginal distribution of $\bx_t$ relative to $\Qt_T(\dx_{0:T})$. This may be used to compute the smoothing expectation of additive functions $\varphi(\bx_{0:T})=\sum_{t=0}^T\varphi_t(\bx_t)$ such as, e.g., the score functions of certain models \citep[e.g.][]{Poyiadjis2011}. See Section \ref{subsec:BS_Marg}. 

The second type of backward step \citep{GodsDoucWest} allows to estimate the full (joint) smoothing distribution $\Qt(\dx_{0:T})$. Its SQMC version is  given  and analysed in Section \ref{subsec:BS}.

These two algorithms share the following properties: 
(a) they require that the Markov kernel $m_t(\bx_{t-1},\dx_t)$ admits a positive probability density $m_t(\bx_{t-1},\bx_t)$ which may be computed pointwise (for 
all $\bx_{t-1},\bx_t\in\setX$);
(b) they use as input the output of a forward pass, i.e. either Algorithm \ref{alg:SMC} (SMC),
or  Algorithm \ref{alg:SQMC_B} (SQMC);
and (c) their complexity is $\bigO(TN^2)$. 


\subsection{Marginal backward smoothing}\label{subsec:BS_Marg}

To perform marginal smoothing, it suffices to compute, from the output of the forward pass, the following smoothing weights:
\begin{equation*}
\widetilde{W}_{t|T}^n \eqdef 
W_t^n \times \sum_{m=1}^{N} 
\frac{\widetilde{W}_{t+1|T}^m m_{t+1}(\bx_t^n,\bx_{t+1}^m)G_{t+1}(\bx_t^n,\bx_{t+1}^m)}
{\sum_{p=1}^{N} W_t^p m_{t+1}(\bx_t^p,\bx_{t+1}^m)G_{t+1}(\bx_t^p,\bx_{t+1}^m)}
\end{equation*}
for all $n\in 1:N$, and recursively, from $t=T-1$, to $t=0$. (For $t=T$, simply 
set $\widetilde{W}_{t|T}^n = W_T^n$.) Then 
\begin{equation*}
\Qt_{t|T}^{N}(\dx_t)\eqdef\sum_{n=1}^N \widetilde{W}_{t|T}^n \delta_{\bx_t^n}(\dx_t) \approx \Q_{t|T}(\dx_t).
\end{equation*}

This particular backward pass may be applied to either the output of SMC (Algorithm \ref{alg:SMC}), or SQMC (Algorithm \ref{alg:SQMC_B}). In the latter case, the question is whether this approach remains valid.  The answer is directly given by 
Theorem \ref{thm:BS}: under its assumptions, we have that 
  \begin{align*}
    \|\Qt_{t|T}^{N}-\Qt_{t|T}\stn\cvz,\quad\mbox{as }N\rightarrow+\infty
  \end{align*}
since $\Qt_{t|T}$ (resp. $\Qt_{t|T}^{N}$) is a certain marginal distribution of $\Qt_T$
(resp. $\Qt_T^N$). In words, marginal backward smoothing generates consistent 
(marginal) smoothing estimates when applied to the output of the SQMC algorithm.

\subsection{Full backward smoothing}\label{subsec:BS}

The SQMC backward step to estimate the joint smoothing distribution $\Qt_T$, proposed in \citet{SQMC}, is recalled as Algorithm \ref{alg:back}.

\begin{algorithm}[h]
  \caption{SQMC Backward step for full smoothing\label{alg:back}}
  \begin{algorithmic}
	\Require{$\bx_{t}^{\sigma_{t}(1:N)}$,
      $W_{t}^{\sigma_{t}(1:N)}$ for $t\in 0:T$, output of Algorithm
      \ref{alg:SQMC_B} after the Hilbert sort step (i.e, for all $n,m\in 1:N$,
       $n\leq m\implies h (\bx_t^{\sigma_t(n)})\leq h(\bx_t^{\sigma_t(m)})$) and $\tilde{\bu}^{1:N}$ a point set in $\ui^{T+1}$ such that, for all $n,m\in 1:N$, $n\leq m\implies u_T^n\leq u_T^m$}
	
	\Ensure $\tilde{\bx}^{1:N}_{0:T}$ ($N$ trajectories in
    $\setX^{T+1}$)

    \For{$n=1\rightarrow N$}
    \State Compute $\tilde{\bx}_{T}^{n}=\bx_{T}^{a_{T}^{n}}$ where $a_{T}^n=F_{T,N}^{-1}(u_T^{n})$ with 
    $F_{T,N}(i)=\sum_{m=1}^NW^{\sigma_T(m)}_{T}\mathbb{I}(m\leq i)$
    \EndFor

    \For{$t=T-1\rightarrow 0$}
     \For{$n=1\rightarrow N$} 
    \State\label{state:EmpCDF}  Compute $\tilde{\bx}_t^{n}=\bx_t^{\tilde{a}_t^n}$ where $\tilde{a}_t^n=\tilde{F}_{t,N}^{-1}(\tilde{\bx}^n_{1+1},\tilde{u}_{t}^{n})$ with 
$
     \tilde{F}_{t,N}(\bx_{t+1}, i)=\sum_{m=1}^N
\widetilde{W}^{\sigma_t(m)}_{t}(\bx_{t+1})\mathbb{I}(m\leq
i), 
$ 
and $\widetilde{W}_{t}^{m}(\bx_{t+1})=
\frac{W_t^m m_{t+1}(\bx_t^m,\bx_{t+1})G_{t+1}(\bx_t^m,\bx_{t+1})}
{\sum_{p=1}^{N}W_t^p m_{t+1}(\bx_t^p,\bx_{t+1})G_{t+1}(\bx_t^p,\bx_{t+1})}.
$
\EndFor
\EndFor
\end{algorithmic}
\end{algorithm}

Algorithm \ref{alg:back} generates a low discrepancy point set for distribution 
$\Qt_{T}^{N}$, the plug-in estimate of $\Qt_{T}$, and is therefore the exact QMC equivalent of the backward
step of standard backward sampling.

To better understand why Algorithm \ref{alg:back} is valid, it helps to decompose 
it in two steps. First, it transforms $\tilde{\bu}^{1:N}$, a  point
set in $\ui^{T+1}$, into 
$\tilde{h}_{0:T}^{1:N}$, another point set in $[0,1)^ {T+1}$, by applying 
the inverse Rosenblatt transformation of 
\begin{equation}\label{eq:RSback}
  \Qt_{T,h}^{N}(\dd h_{0:T}):=\Qh^N_{T,h}(\dd h_T)\prod_{t=1}^T
  \mathcal{M}^h_{t,\Qh^N_{t-1,h}}(h_t,\dd h_{t-1}),
\end{equation}
which is the image of  probability measure
$\Qt_{T}^{N}(\dx_{0:T})$, defined in \eqref{eq:backwardN}, by mapping 
$h_T: (\bx_0,\dots,\bx_T)\mapsto(h(\bx_0),\dots,h(\bx_T))$. 
Recall that $\Qh_{t,h}^N$ is the image of $\Qh_t^N$ by $h$
while, for any $\pi\in\mathcal{P}(\ui)$ and $t\in 1:T$,
$\mathcal{M}^h_{t+1,\pi}:\ui\rightarrow \mathcal{P}(\ui)$ is a Markov
kernel such that
$$ \mathcal{M}^h_{t,\pi}(h_{t},\dd
h_{t-1})
\propto m_{t}\big(H(h_{t-1}),H(h_{t})\big)G_t\big( H(h_{t-1}),H(h_t)\big)
\pi(\dd h_{t-1}).
$$

In a second step, Algorithm \ref{alg:back} returns $\tilde{\bx}^{1:N}_{0:T}$
where $\tilde{\bx}^{n}_{0:T}=H_T(\tilde{h}_{0:T}^{n})$ with the mapping $H_T:\ui^{T+1}\rightarrow \ui^{d(T+1)}$ defined in \eqref{eq:HT}.

\subsubsection{$L_1-$ and $L_2-$convergence}\label{subsec:L1L2BS}

A direct consequence of the inverse Rosenblatt interpretation of the previous section is that, when Algorithm \ref{alg:back} uses a RQMC point set as input, the random point $\tilde{\bx}_{0:T}^n$ is such that, for any function $\varphi:\ui^{d(T+1)}\rightarrow\mathbb{R}$ and for any $n\in 1:N$, we have  $\E[\varphi(\tilde{\bx}_{0:T}^n)|\mathcal{F}_T^N]=\Qt_T^N(\varphi)$, with $\mathcal{F}_T^N$  the $\sigma$-algebra generated by the forward step. Together with Theorem \ref{thm:BS}, this observation allows us to deduce $L_2$-convergence for test functions $\varphi$ that are continuous and bounded (see Appendix \ref{p-thm:L2} for a proof).

\begin{thm}\label{thm:L2}
Consider the set-up of the SQMC forward filtering-backward smoothing algorithm (Algorithms \ref{alg:SQMC_B} and \ref{alg:back}) and assume the following:
\begin{enumerate}
\item\label{H:u1} In Algorithm \ref{alg:SQMC_B}, $(\bu_{t}^{1:N})_{N\geq 1}$, $t\in0:T$, are independent random   sequences of  point sets in
$\ui^{d_t}$, with $d_0=d$ and $d_t=d+1$ for $t>0$, such
that, for any $\epsilon>0$, there exists a $N_{\epsilon,t}>0$
such that, almost surely, $D(\bu_{t}^{1:N})\leq\epsilon$, $\forall N\geq N_{\epsilon,t}$;

\item\label{H:u2} In Algorithm \ref{alg:back}, $(\tilde{\bu}^{1:N})_{N\geq 1}$ is a sequence of point sets in $\ui^{T+1}$ such that 
\begin{enumerate}
\item  $\forall n\in 1:N$, $\tilde{\bu}^n\sim\Unif(\ui^{T+1})$;
\item For any function $\varphi\in L_{2}\left(\ui^{d(T+1)},\lambda_{d_{t}}\right)$,
$
\var\left(\frac{1}{N}\sum_{n=1}^{N}\varphi(\bu_{t}^{n})\right)\leq C\sigma_{\varphi}^{2}r(N)
$
where  $\sigma_{\varphi}^{2}=\int\left\{ \varphi(\bu)-\int\varphi(\bv)\dd \bv\right\} ^{2}\du$, $r(N)\rightarrow 0$ as $N\rightarrow +\infty$, 
and where both $C$ and $r(N)$ do not depend on $\varphi$;
\end{enumerate}    
\item Assumptions of Theorem \ref{thm:consistency2} and Assumptions H1-H2 of Theorem \ref{thm:BS} hold. 
\end{enumerate}
Then, for any continuous and bounded function $\varphi:\setX^{T+1}\rightarrow\mathbb{R}$,
$$
\E \left|\Sop(\tilde{\bx}_{0:T}^{1:N})(\varphi)-\Qt_T(\varphi)\right|\cvz,\quad \var\left(\Sop(\tilde{\bx}_{0:T}^{1:N})(\varphi)\right)\cvz,\quad\mbox{as }N\rightarrow+\infty.
$$
\end{thm}

Assumption \ref{H:u1} is  verified for instance when  $\bu_t^{1:N}$ consists of the first $N$ points of a  nested scrambled $(t,d_t)$-sequence  in base $b\geq 2$ \citep{Owen1995,Owen1997a,Owen1998}. The result above may be easily extended to the case  where the  $\bu_t^{1:N}$'s are deterministic (rather than random) QMC point sets.

On the other hand, the point set $\tilde{\bu}^{1:N}$ used as input of the backward pass is necessarily random (for the result above to hold). But $\tilde{\bu}^{1:N}$ does not need to be a QMC point set (i.e. to have low discrepancy). In particular, Assumption \ref{H:u2} is satisfied when the $\tilde{\bu}^{1:N}$ are IID uniform variates (in $\ui^{T+1}$); then $C=1$ and $r(N)=N^{-1}$. See Section \ref{sec:altBack} for a discussion one the use of QMC or pseudo-random numbers in the backward step of SQMC.

\subsubsection{Consistency}\label{sec:3.3}

Compared to standard (forward) SQMC, establishing the consistency 
of SQMC backward smoothing requires two extra technical steps. 
First, as Algorithm \ref{alg:back} generates a point set
$\tilde{h}_{0:T}^{1:N}$ in $\ui^{T+1}$ using the inverse Rosenblatt
transformation of the probability measure defined in \eqref{eq:RSback}, and then projects it back to 
$\setX^{T+1}$ through $H_T$, we need to establish
that this transformation preserves the low discrepancy properties of
$\tilde{h}_{0:T}^{1:N}$. For this we will use 
Theorem \ref{thm:Hilbert2}.

Second, the proof of \cite{SQMC} for the consistency of SQMC 
required  smoothness conditions on the Rosenblatt transformation of $m_{t,h}(h_{t-1},\dx_t)=m_{t}(H(h_{t-1}),\dx_t)$, so that this transformation
maintains low discrepancy, as explained in Section \ref{subsec:sqmc}.  
Due to the H\"{o}lder property of the Hilbert curve, the H\"{o}lder
continuity of $F_{m_t}$ implies that $F_{m_{t,h}}$ is H\"{o}lder
continuous as well. Similarly, for the backward step we need
assumptions on the Markov kernel $\mathcal{M}_{t,\Q_{t-1}}$ which
imply sufficient smoothness for the Rosenblatt transformation of
$\mathcal{M}_{t,\Qh^N_{t-1,h}}^h$ which is used in the course of Algorithm \ref{alg:back} to transform the QMC point set in $\ui^{T+1}$.

To this aim, note that since $\|\Qh^N_{t-1}-\Q_{t-1}\stn\rightarrow 0$ as $N\rightarrow+\infty$ (Theorem \ref{thm:consistency2}),  one may expect that 
$$
\|\mathcal{M}_{t,\Qh^N_{t-1,h}}^h-\mathcal{M}_{t,\Q_{t-1,h}}^h\stn\cvz,\quad \text{as $N\rightarrow +\infty$}.
$$
Therefore, we intuitively need  smoothness assumption on this limiting Markov kernel to establish the
validity of the backward pass of SQMC. However, note that the two
arguments of this kernel are ``projections'' in $\ui$ through the
inverse of the Hilbert curve. Consequently, it is not clear how smoothness assumptions on the Rosenblatt transformation of $\mathcal{M}_{t,\Q_{t-1}}$ would translate into some regularity for the Rosenblatt transformation of $\mathcal{M}_{t,\Q_{t-1,h}}^h$. As shown below, a consistency result for
 QMC forward-backward algorithm can be established under a H\"{o}lder
assumption on the CDF of $\mathcal{M}_{t,\Q_{t-1}}$.

To establish the consistency of Algorithm \ref{alg:back}  we proceed in two steps. First, we consider a modified backward pass which amounts to sampling from a
continuous distribution. Working with a continuous distribution allows us to focus on the  technical difficulties specific to the backward step we just mentioned without being 
 distracted by complicated discontinuity issues. Then, the result  obtained for this continuous backward pass is used to deduce sufficient conditions for the consistency of Algorithm \ref{alg:back}. If this approach in two steps greatly facilitates the analysis,  the resulting conditions for the validity of QMC forward-backward smoothing have the drawback to impose that the Markov kernel $m_t$ and the potential function $G_t$ are  bounded below away from zero (see Corollary \ref{cor:BS_star} below).
 
\subsubsection{A continuous backward pass}

Following the discussion above, we consider now a modified backward pass, which amounts to transforming a QMC point set
$\tilde{\bu}^{1:N}$ in $\ui^{T+1}$ through the inverse Rosenblatt
transformation of a continuous approximation
$\widetilde{\mathsf{Q}}_{T,h_T}^N$ of $\Qt_{T,h_T}^N$. 

To construct
$\widetilde{\mathsf{Q}}_{T,h_T}^N$, let $\widehat{\mathsf{Q}}_{T,h}^N$
be the probability measure that corresponds to a continuous
approximation of the CDF of $\Qh^N_{T,h}$, which is strictly increasing
on $[0,h(\bx_T^{\sigma_T(N)})]$ with
$F_{\widehat{\mathsf{Q}}_{T,h}^N}(h(\bx_T^{\sigma_T(N)}))=1$ and such
that, under the assumptions of Theorems \ref{thm:consistency2} and
\ref{thm:BS},
$$ \|\widehat{\mathsf{Q}}_{T,h}^N-\Qh_{T,h}^{N}\stn=\smallo(1).
$$ Next, for $t\in 1:T$, let $K^N_{t,h}:\ui\rightarrow\mathcal{P}(\ui)$ be a Markov kernel such that:

\begin{enumerate}
\item Its CDF is continuous on $\ui\times
[0,h(\bx_{t-1}^{\sigma_{t-1}(N)})]$;

\item $\forall h_t\in\ui$, the CDF
of $K^N_{t,h}(h_t,\dd h_{t-1})$ is strictly increasing on $[0,
h(\bx_{t-1}^{\sigma_{t-1}(N)})]$ with
$F_{K^N_{t,h}}(h_t,h(\bx_{t-1}^{\sigma_{t-1}(N)}))=1$;

\item Under the assumptions of Theorems \ref{thm:consistency2} and and
\ref{thm:BS}, 
$$
 \sup_{h_t\in \ui}\|K^N_{t,h}(h_t,\dd
h_{t-1})-\mathcal{M}^{h}_{t,\Qh^N_{t-1,h}}(h_t,\dd
h_{t-1})\stn=\smallo(1).
$$
\end{enumerate}

Finally, we define $\widetilde{\mathsf{Q}}_{T,h_T}^N\in\mathcal{P}(\ui^{T+1})$
as
\begin{align*}
  \widetilde{\mathsf{Q}}_{T,h_T}^N(\dd
  h_{0:T}):=\widehat{\mathsf{Q}}_{T,h}^N(\dd h_T)\prod_{t=1}^T
  K^N_{t,h}(h_t,\dd h_{t-1})
\end{align*}
which, by construction, has a Rosenblatt transformation which is continuous on $\ui^{T+1}$.

Remark that such a distribution
$\widetilde{\mathsf{Q}}_{T,h_T}^N$ indeed exists. For instance,
under the assumptions of Theorems \ref{thm:consistency2} and
\ref{thm:BS}, one can take for $\widehat{\mathsf{Q}}_{T,h}^N$ the
probability distribution that corresponds to a piecewise linear approximation of
the CDF of $\Qh^N_{T,h}$ and, similarly, for $h_t\in\ui$, one can
construct $K^N_{t,h}(h_t,\dd h_{t-1})$ from a piecewise linear approximation of
the CDF of $\mathcal{M}^{h}_{t,\Qh^N_{t-1,h}}(h_t,\dd h_{t-1})$.

For this modified backward step we  obtain the following consistency result:
 
\begin{thm}\label{thm:BS_star}
  Let $(\tilde{\bu}^{1:N})_{N\geq 1}$ be a sequence of point sets in
  $\ui^{T+1}$ such that $D(\tilde{\bu}^{1:N})\cvz$ as
  $N\rightarrow+\infty$.  For $n\in 1:N$, let
  $\check{h}_{0:T}^{n}=F^{-1}_{\widetilde{\mathsf{Q}}_{T,h_T}^N}(\tilde{\bu}^{n})$
  where $\widetilde{\mathsf{Q}}_{T,h_T}^N$ is as above. Suppose that
  the Assumptions of Theorem \ref{thm:consistency2} and Assumptions H1-H2 of Theorem \ref{thm:BS} hold and that, viewed
  as a function of $\bx_t$ and $\bx_{t-1}$,
  $F_{\mathcal{M}_{t,\Q_{t-1}}}^{cdf}(\bx_t,\bx_{t-1})$, the CDF of
  $\mathcal{M}_{t,\Q_{t-1}}(\bx_t,\dx_{t-1})$, is H\"{o}lder
  continuous for all $t\in 1:T$. Let
  $\check{\bx}_{0:T}^n=H_T(\check{h}_{0:T}^{n})$. Then,
$$ \|\Sop(\check{\bx}_{0:T}^{1:N})-\Qt_T\stn\cvz\quad\text{as
}N\rightarrow+\infty.
$$
\end{thm}
See Appendix \ref{p-thm:BS_star} for a proof.

\subsubsection{A consistency result for SQMC forward-backward smoothing}

We are now ready to provide conditions which ensure that QMC
forward-backward smoothing (Algorithms \ref{alg:SQMC_B} and
\ref{alg:back}) yields a consistent estimate of the smoothing
distribution. The key idea of our consistency result (Corollary
\ref{cor:BS_star} below) is to show that, for a given point set
$\tilde{\bu}^{1:N}$, the point set $\tilde{\bx}_{0:T}^{1:N}$ generated
by Algorithm \ref{alg:back} becomes, as $N$ increases, arbitrary
close to the point set $\check{\bx}_{0:T}^{1:N}$ obtained by the
modified backward step described in the previous subsection.

\begin{corollary}\label{cor:BS_star}
  Consider the set-up of the SQMC forward filtering-backward smoothing
  algorithm (Algorithms \ref{alg:SQMC_B} and \ref{alg:back}) and
  assume the following holds for $t\in 0:T-1$:
  \begin{enumerate}
  \item $(\tilde{\bu}^{1:N})_{N\geq 1}$ is a sequence of point sets in
    $\ui^{T+1}$ such that $D(\tilde{\bu}^{1:N})\cvz$ as
    $N\rightarrow+\infty$;
  \item Assumptions of Theorem \ref{thm:consistency2} and Assumptions H1-H2 of Theorem \ref{thm:BS} hold;
  \item 
    $F_{\mathcal{M}_{t,\Q_{t-1}}}^{cdf}(\bx_t,\bx_{t-1})$ is
    H\"{o}lder continuous;
  \item \label{H:corr:gamma}There exists a constant $\underline{c}_t>0$ such that, for all $\bx_{(t-1):(t+1)}\in\setX^3$,
$$ 
G_t(\bx_{t-1},\bx_t)G_{t+1}(\bx_{t},\bx_{t+1})m_{t+1}(\bx_{t},\bx_{t+1})\geq \underline{c}_t;
$$
\item\label{H:corr:Unif}  $G_t(\bx_{t-1},\bx_t)m_t(\bx_{t-1},\bx_t)$ is uniformly  continuous on $\setX^2$.
\end{enumerate}
Then,
$$ \|\Sop(\tilde{\bx}_{0:T}^{1:N})-\Qt_T\stn\cvz\quad\text{as
}N\rightarrow+\infty.
$$
\end{corollary}

See Appendix \ref{p-cor:BS_star} for a proof. Recall that the result above implies
that 
$$ \frac{1}{N} \sum_{n=1}^{N} \varphi(\tilde{\bx}_t^n) \rightarrow \Qt_T(\varphi),
\quad \mbox{as } N\rightarrow +\infty$$
for any bounded and continuous $\varphi$, as explained in Section \ref{sub:QMC}. 

Assumption \ref{H:corr:gamma} is the main assumption of this result. This strong condition is the price to pay for our study of QMC backward smoothing in two steps which, again, has the advantage to facilitate the analysis by avoiding complicated discontinuity problems. 
We conjecture that this assumption may be removed by using an approach
similar to the proof of Theorem 4 in \cite{SQMC}.

\subsection{An alternative backward step}\label{sec:altBack}

A  drawback of Algorithm \ref{alg:back} is that it uses as an input 
a point set of $\tilde{\bu}^{1:N}$ of dimension $(T+1)$, although $T$ is often
large in practice. It is well known that high-dimensional QMC point sets do
not have good equidistribution properties, unless $N$ is extremely large. 

To address this issue, we may still use SQMC for the forward pass, but use 
as a backward pass Algorithm \ref{alg:back} with IID uniform variables as an input
(i.e. input $\tilde{\bu}^{1:N}$ is replaced by $N$ uniforms). 
Our consistency results still apply,
since $D(\tilde{\bu}^{1:N})\cvz$ with probability one in that case \citep[][page 167]{Niederreiter1992}.
Of course, one cannot hope for a convergence rate better than $N^{-1/2}$
for such a hybrid approach, but the resulting algorithm may still perform better than 
standard (Monte Carlo) backward smoothing (for fixed $N$), while being simpler to implement
than SQMC with a QMC backward pass based on a point set of dimension $T+1$. 

More generally, we could take  $\tilde{\bu}^{1:N}$ to be some combination of a point
sets and uniform variables, while still having $D(\tilde{\bu}^{1:N})\cvz$
\citep{Okten2006}.
However, we leave for further research the study of such an extension.

\section{Numerical study}\label{sec:num}

We consider the following multivariate stochastic
volatility model (SV) proposed by \citet{Chan2006}:
\begin{equation}\label{simu:eq:modelSV_B}
  \begin{cases}
    \by_{t}=S_t^{1/2}\boldsymbol{\epsilon}_{t},\quad &t\geq 0\\
    \bx_{t}=\boldsymbol{\mu}+\Phi(\bx_{t-1}-\boldsymbol{\mu})+
    \Psi^{\frac{1}{2}}\boldsymbol{\nu}_t,\quad &t\geq 1
  \end{cases}
\end{equation}
where $S_t=\text{diag}(\exp(x_{t1}),...,\exp(x_{td}))$, $\Phi$ and
$\Psi$ are diagonal matrices and
$(\boldsymbol{\epsilon}_{t},\boldsymbol{\nu}_{t})\sim\mathcal{N}_{2d}(\bm{0}_{2d},
C)$ with $C$ a correlation matrix.

The parameters we use for the simulations are the same as in
\citet{Chan2006}: 
$\phi_{ii}=0.9$, $\mu_i=-9$, $\psi_{ii}=0.1$ for all $i=1,...,d$ and
$$ C=
\begin{pmatrix}
  0.6\mathbf{1}_d+0.4\mathbf{I}_d&\bm{0}_{d}\\
  \bm{0}_{d}&0.8\mathbf{1}_d+0.2\mathbf{I}_d
\end{pmatrix}
$$ 
where $\mathbf{I}_d$, $\bm{0}_{d}$ and
$\mathbf{1}_d$, are respectively the identity, all-zeros, and all-ones
 $d\times d$ matrices. The prior distribution
for $\bx_0$ is the stationary distribution of the process
$(\bx_t)_{t\geq 0}$. We take $d=2$ and
$T=399$ (i.e. 400 observations).

We report results (a) for QMC full backward smoothing
(Algorithm \ref{alg:SQMC_B} for the forward pass, then Algorithm \ref{alg:back} for the backward pass),  and (b)  for marginal backward smoothing (as described
in Section \ref{subsec:BS_Marg}). These algorithms are compared with their Monte Carlo
counterpart using the gain factors for the estimation of the smoothing
expectation $\E[x_{1t}|\by_{0:T}]$, $t\in 0:T$, which we define as the
Monte Carlo mean square error (MSE) over the quasi-Monte Carlo
MSE. Results for component $x_{2t}$ of $\bx_t$ are mostly similar (by symetry)
and thus are not reported. 

The implementation of QMC and Monte Carlo algorithms are as in
\citet{SQMC}. 
In SQMC, prior to the Hilbert sort step, the particles are mapped into $\ui^d$ using a
component-wise (rescaled) logistic transform. For SMC, systematic
resampling \citep{CarClifFearn} is used, and random
variables are generated using standard methods (i.e. not using the
inverse Rosenblatt transformation). The complete C/C++ code is
available on-line at \url{https://bitbucket.org/mgerber/sqmc}.

\begin{figure}
  \begin{centering}
    \begin{tabular}{cc}
      $N=2^8$&$N=2^{10}$\\ \includegraphics[scale=0.35]{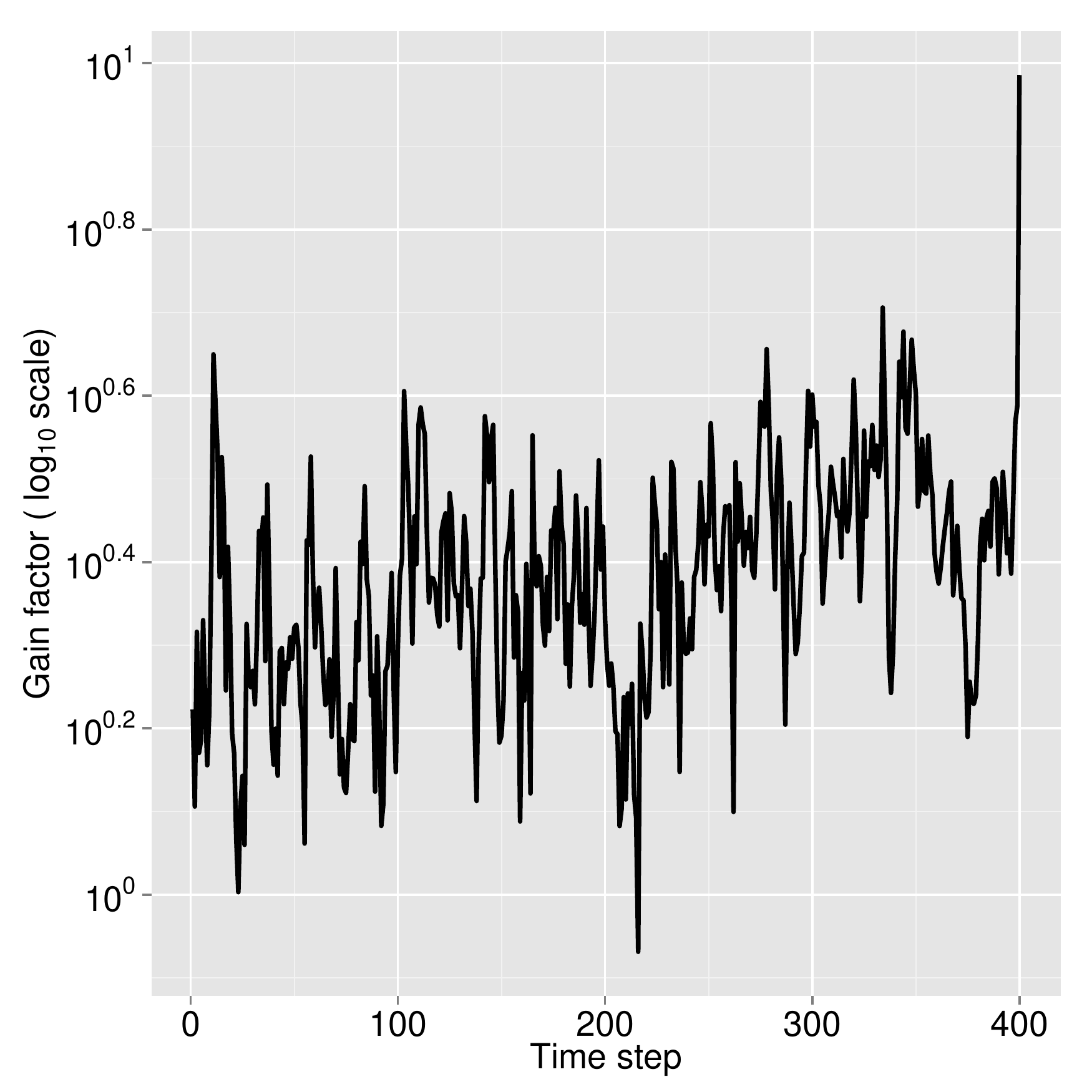}&\includegraphics[scale=0.35]{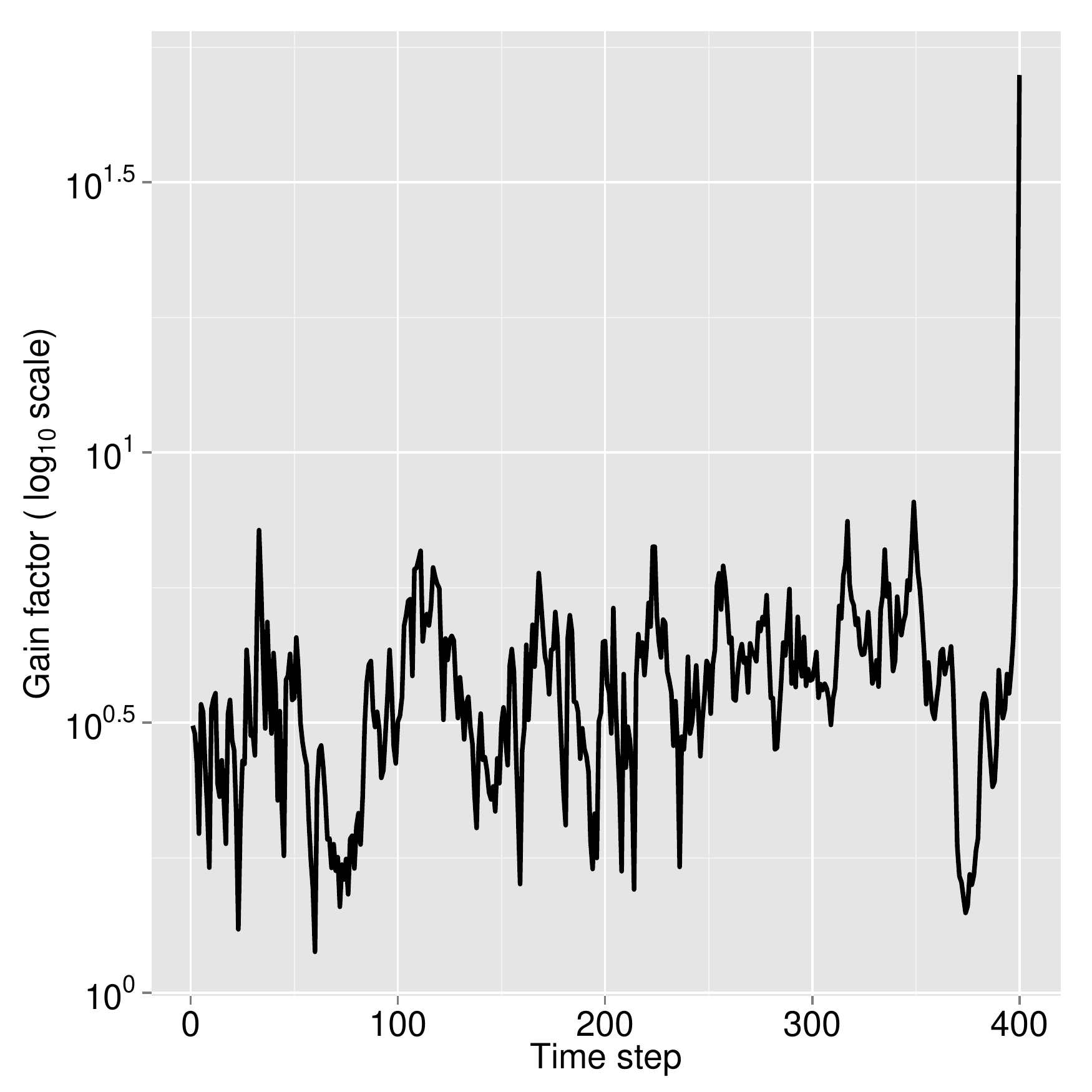}\\
      \multicolumn{2}{c}{Full backward smoothing} \\
       \includegraphics[scale=0.35]{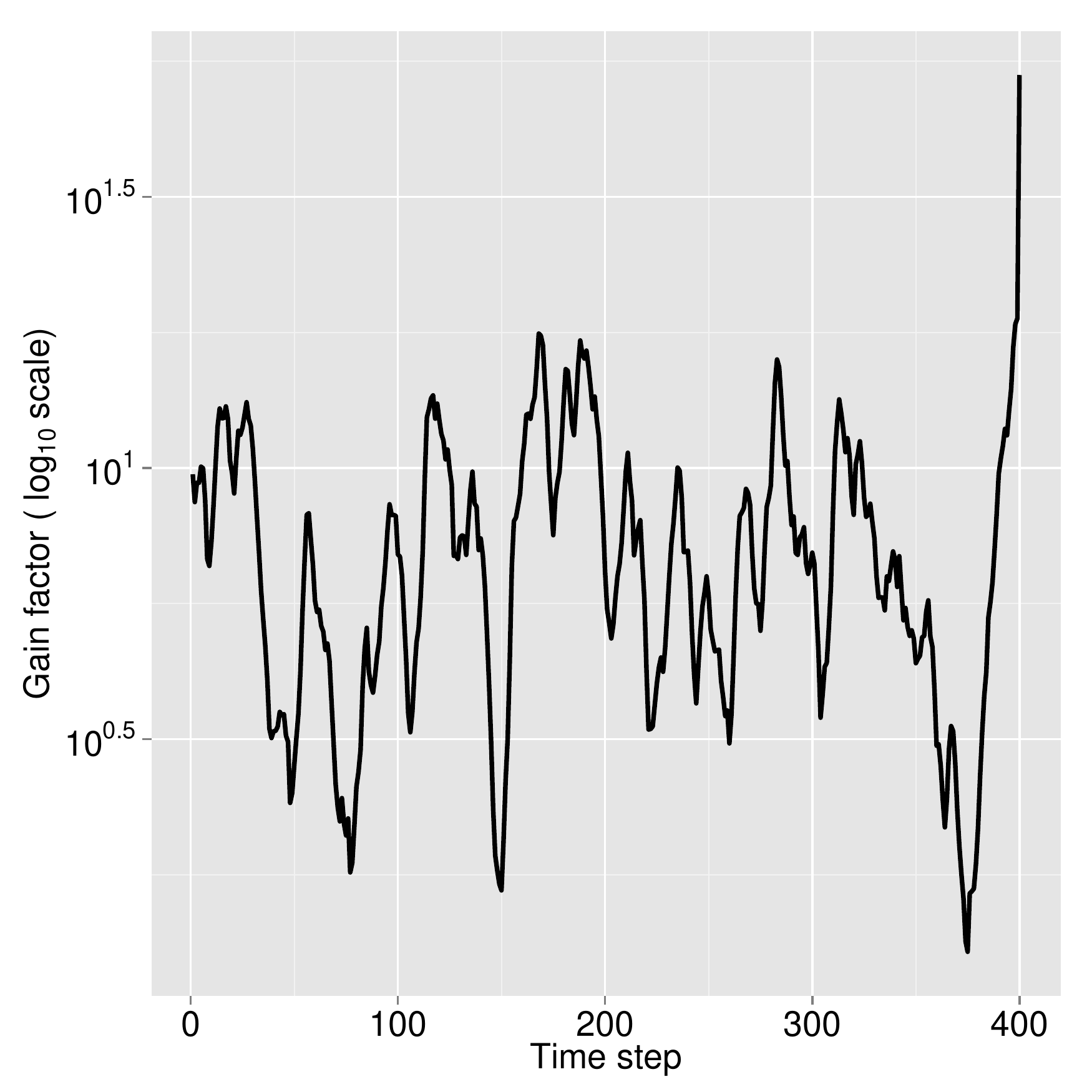}&\includegraphics[scale=0.35]{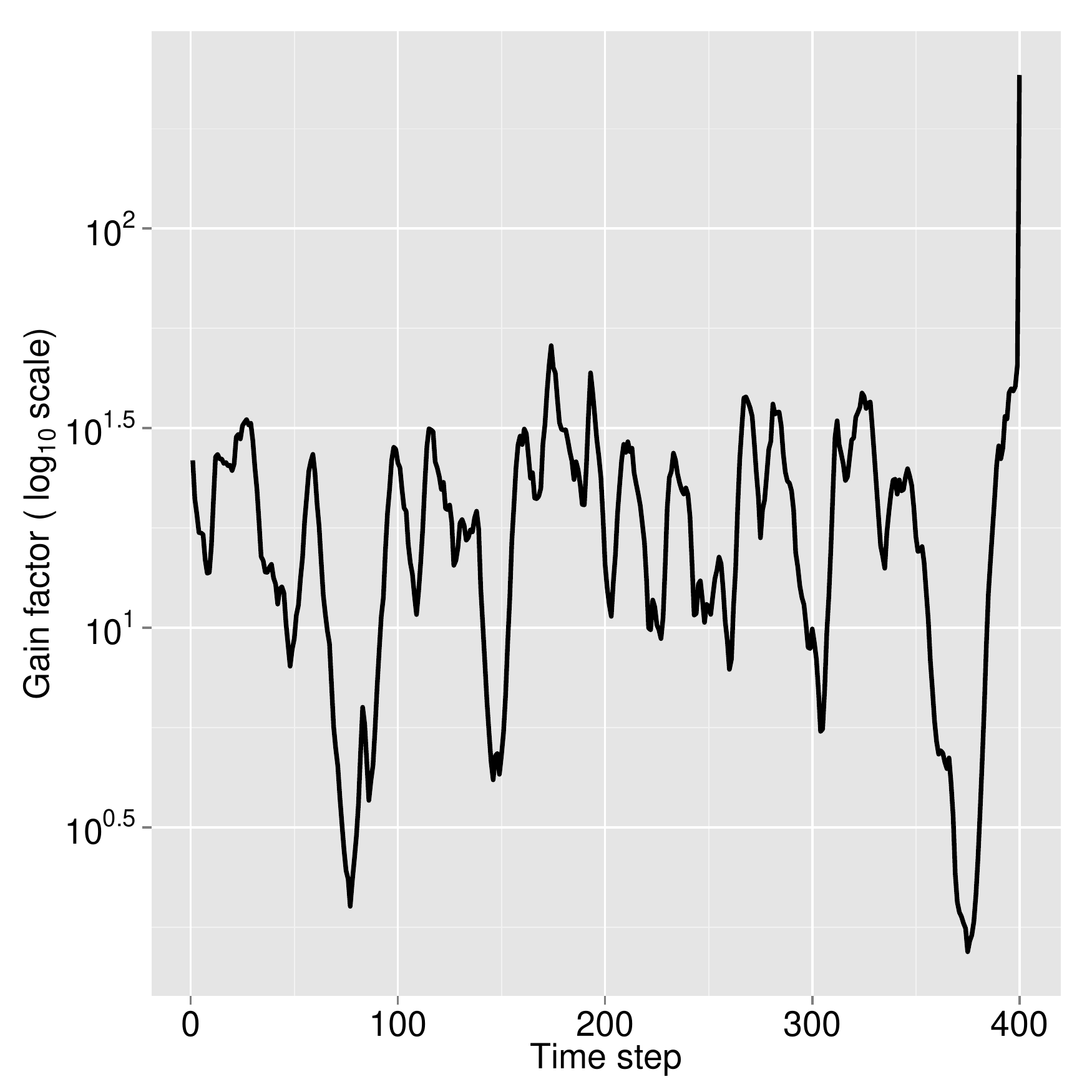}\\ 
         \multicolumn{2}{c}{Marginal backward smoothing} \\
    \end{tabular}
    \par\end{centering}
  \caption{Smoothing of the bivariate SV model
    \eqref{simu:eq:modelSV_B} for $N=2^8$ and $N=2^{10}$ particles. The
    graphs give the gain factor (MSE ratio, from 100 replications) for comparing SQMC with SMC, and for 
     $\E[x_{t1}|\by_{0:T}]$ as a function of $t$. The top line is for
full backward smoothing (Algorithm \ref{alg:back}), the bottom line is for marginal backward smoothing. \label{fig:SV2}}
\end{figure}

\begin{figure}
  \begin{centering}
    \begin{tabular}{cc}
      $N=2^8$&$N=2^{10}$\\ \includegraphics[scale=0.35]{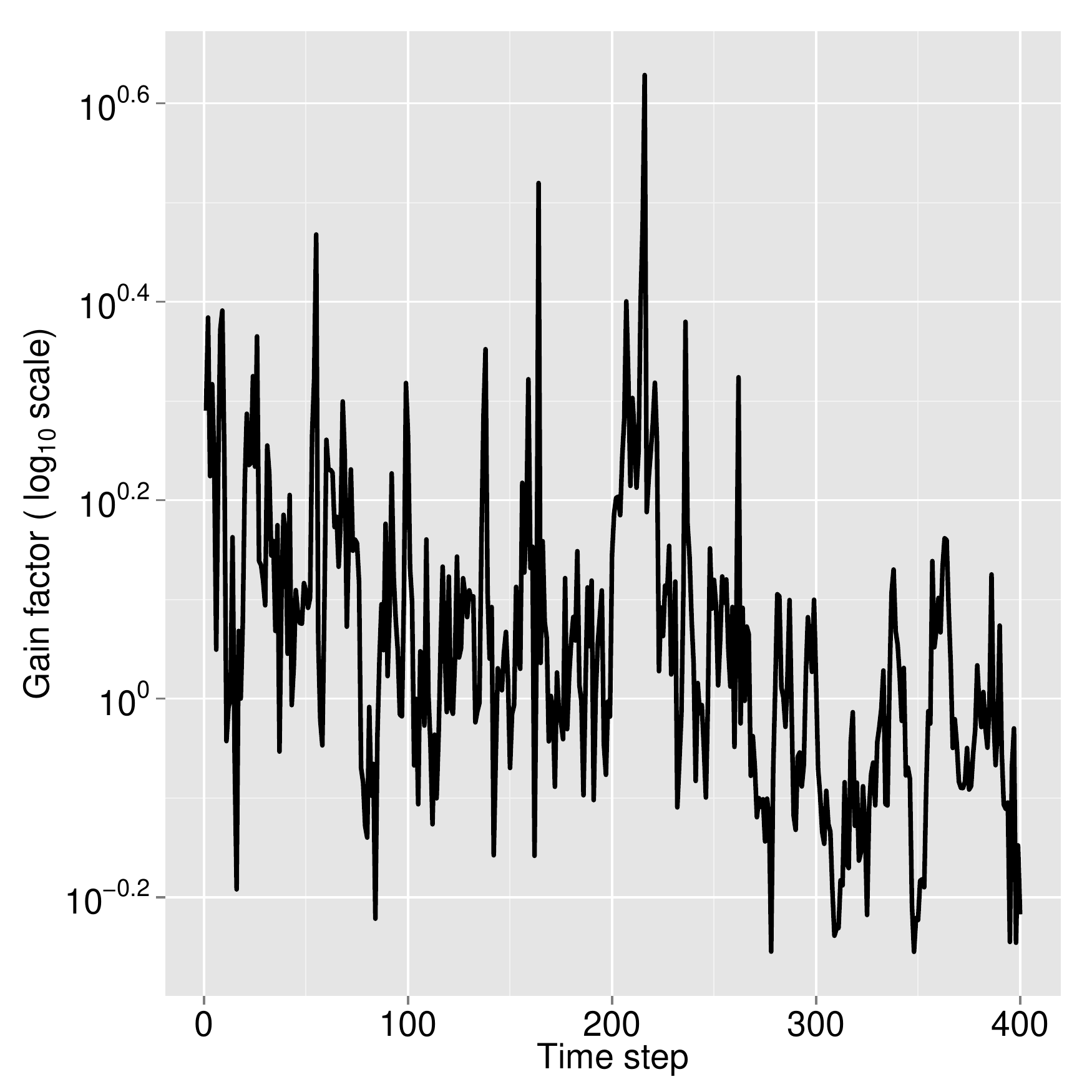}&\includegraphics[scale=0.35]{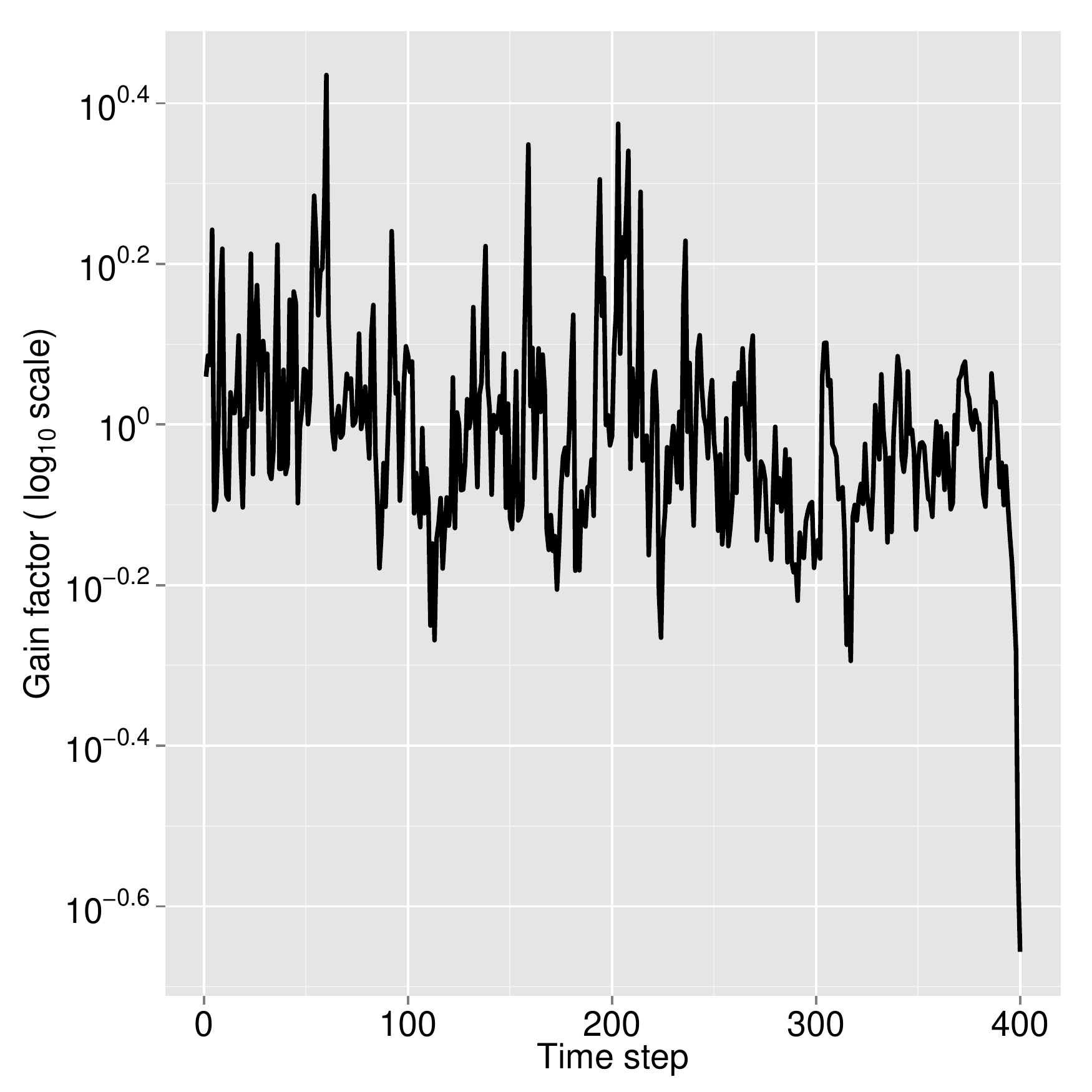}\\
    \end{tabular}
    \par\end{centering}
  \caption{Smoothing of the bivariate SV model
    \eqref{simu:eq:modelSV_B} for $N=2^8$ and $N=2^{10}$ particles. 
    The graphs give the gain factor (MSE ratio, for 100 replications) of the hybrid backward pass (Algorithm \ref{alg:back} with IID input) 
relative to the QMC backward pass (Algorithm \ref{alg:back} with a QMC point set as input), for the estimation of $\E[x_{t1}|\by_{0:T}]$ as a
    function of $t$.  \label{fig:SV2:Hybrid}}
\end{figure}

Figure \ref{fig:SV2} plots the gain factors at each time step, for either
$N=2^8$ (left), or $N=2^{10}$ (right). We observe that gain factors
tend to increase with $N$ (as expected) and that they are above one 
most of the time. They are not very high for full backward smoothing; but note
that even a marginal improvement in terms of gain factor may translate in high
CPU time savings, given that these algorithms have complexity $\bigO(N^2)$;
i.e. a gain factor of 3 means that SMC would need 3 times more particles, and therefore 9 times more CPU time, to reach the same accuracy as SQMC. 
Notice also gain factors are higher for marginal smoothing.

Finally, we compare Algorithm \ref{alg:back} (full backward smoothing)
with the hybrid strategy described at the end of Section \ref{sec:altBack}: i.e.
a SQMC forward pass (Algorithm \ref{alg:SQMC_B}) followed by a Monte Carlo
backward pass. Again, this is for  $N=2^8$ (left) and $N=2^{10}$ (right).
Interestingly, the hybrid strategy (slightly) dominates at most time steps (excepts
those such that $T-t$ is small). As already discussed, the likely reason
for this phenomenon is that the backward pass of Algorithm \ref{alg:back} is
based on a point set of dimension $T$, which is too large to have 
good equidistribution properties (for reasonable values of $N$), and therefore
to bring much improvement over plain Monte Carlo. Thus, for large $T$, 
one may as well use this hybrid strategy to perform full smoothing.

\section{Conclusion}\label{sec:conclusion}

The estimation of the smoothing distribution
$p(\bx_{0:T}|\by_{0:T})$ is a challenging task for QMC methods because
it is typically a high dimensional problem. On the other hand, due to
the $\bigO(N^2)$ complexity of most smoothing algorithms, small gains
in term of mean square errors translate into important savings in term
of running times to reach the same level of error. In this work we
provide asymptotic results for some QMC smoothing strategies, namely
forward smoothing, and two variants of forward-backward smoothing.  
In a simulation study we show that the QMC
forward-backward smoothing algorithm outperforms its Monte Carlo
counterpart despite of the high dimensional nature of the
problem. Also, if one is interested in the estimation of the marginal
smoothing distributions, more important gains may be obtained.

The set of smoothing strategies discussed in this work is obviously
not exhaustive. For instance, we have not discussed two-filter smoothing 
\citep{Briers2005}, or its  $\bigO(N)$ variant proposed 
by \citet{Fearnhead2010a}. In fact, our analysis can be easily applied to
derive a QMC version of these algorithms and to provide conditions for
their validity. An other interesting smoothing algorithm is proposed in
\citet{Douc2011}, where the backward pass is an accept-reject
procedure, leading to a $\bigO(N)$ complexity. 
A last interesting smoothing strategy is the
particle Gibbs sampler proposed by \citet{PMCMC} which generates a
Markov chain having the smoothing distribution as stationary
distribution. For these last two methods, the usefulness and the
validity of replacing pseudo-random numbers by QMC point sets remain
interesting open questions.

\section*{Acknowledgements}


We thank Arnaud Doucet, Art B. Owen and Florian Pelgrin for useful
comments. The second author is partially supported by a grant from the
French National Research Agency (ANR) as part of the ``Investissements
d'Avenir'' program (ANR-11-LABEX-0047).

\bibliographystyle{apalike} \bibliography{complete}

\appendix

\section{Main properties of the Hilbert curve}\label{app:Hilbert}

Function $H$ is obtained as the limit of a certain sequence $(H_m)$
of functions $H_m:[0,1]\rightarrow [0,1]^d$ as $m\rightarrow \infty$. 
The proofs of the results presented in this work are based on the
following technical properties of $\HSFC$ and $\HSFC_{m}$. For
$m\geq0$, let $\mathcal{I}_{m}^{d}=\left\{ I_{m}^{d}(k)\right\}
_{k=0}^{2^{md}-1}$ be the collection of consecutive closed intervals
in $[0,1]$ of equal size $2^{-md}$ and such that $\cup
\mathcal{I}_{m}^{d}=[0,1]$. For $k\geq0$,
$S_{m}^{d}(k)=\HSFC_{m}(I_{m}^{d}(k))$ belongs to $\mathcal{S}_m^d$,
the set of the $2^{md}$ closed hypercubes of volume $2^{-md}$ that
covers $[0,1]^d$, $\cup \mathcal{S}_m^d=[0,1]^d$; $S_{m}^{d}(k)$ and
$S_{m}^{d}(k+1)$ are adjacent, i.e. have at least one edge in common
(\emph{adjacency property}).  If we split $I_{m}^{d}(k)$ into the
$2^{d}$ successive closed intervals $I_{m+1}^{d}(k_{i})$,
$k_{i}=2^{d}k+i$ and $i\in 0:2^{d}-1$, then the $S_{m+1}^{d}(k_{i})$'s
are simply the splitting of $S_{m}^{d}(k)$ into $2^{d}$ closed
hypercubes of volume $2^{-d(m+1)}$ (\emph{nesting property}).
Finally, the limit $\HSFC$ of $\HSFC_{m}$ has the \emph{bi-measure
  property}: $\lambda_{1}(A)=\lambda_{d}(\HSFC(A))$, for any
measurable set $A\subset[0,1]$, and satisfies the \textit{H\"older
  condition} $\|H(x_1)-H(x_2)\|\leq C_H|x_1-x_2|^{1/d}$ for all $x_1$
and $x_2$ in $[0,1]$.
For more background on space-filling curves, see 
\cite{Sagan1994}. 

\section{Proofs\label{proof:Back}}

\subsection{Backward decomposition:  Proof of Theorem \ref{thm:BS}}
\label{p-thm:BS}

Lemma 2 of \citet{SQMC} is central for the proof of this result and
is reproduced here for sake of clarity.

\begin{lemma}\label{lemma:Holder_B} 
Let $(\pi^N)_{N\geq 1}$ be a sequence of probability measures
on $\ui^{d_{1}}$ such that $\|\pi^N-\pi\stn\cvz$ as 
$N\rightarrow+\infty$ for some $\pi\in\mathcal{P}(\ui^{d_{1}})$,
 and let $K$ a kernel $\ui^{d_1}\rightarrow\mathcal{P}(\ui^{d_2}))$ such 
that $F_{K}(\bx_{1},\bx_{2})$ is H\"older continuous with its $i$-th
component strictly increasing in $x_{2i}$, $i\in 1:d_2$. Then
\[
\|\pi^N\otimes K-\pi\otimes K\stn=\smallo(1).
\]

\end{lemma}

From Theorem \ref{thm:consistency2}, we know that (for $t\geq1$)
$$
\|\mathcal{S}(P_{t,h}^N)-
\Q_{t-1,h}\otimes m_{t,h}\stn=\smallo(1)
\qquad\mbox{for } P_{t,h}^N=\left(h(\hat{\bx}_{t-1}^{1:N}),\bx_t^{1:N}\right).
$$ 
To establish \eqref{eq:Smooth}, we fix $\bx_{t+1}$, and recognise $\mathcal{M}_{t+1,\Q_{t}}$ as the marginal distribution of $\bx_t$, relative to  joint distribution 
\begin{equation}\label{eq:somejoint}
\frac{\tilde{G}_{t+1}(\bx_t,\bx_{t+1})G_{t,h}(h_{t-1},\bx_t)}
{\Q_{t-1,h}\otimes m_{t,h}(G_{t,h})} \times
\Q_{t-1,h}\otimes m_{t,h}\left(\dd(h_{t-1},\bx_t)\right)
\end{equation}
with $G_{t,h}(h_{t-1},\bx_t)=G_t(H(h_{t-1},\bx_t))$.
This is a change of measure applied to $\Q_{t-1,h}\otimes m_{t,h}$. 
Similarly, $\mathcal{M}_{t+1,\Qh_{t}^N}$ is the marginal of a joint distribution obtained by the same change of measure, but applied to 
$\Sop(P_{t,h}^N)$.  

Thus, we may apply Theorem 1 of \cite{SQMC}, 
and deduce that (again for a fixed $\bx_{t+1}$): 
$$
\|\mathcal{M}_{t+1,\Qh^N_{t}}(\bx_{t+1},\dx_{t})-\mathcal{M}_{t+1,\Q_{t}}(\bx_{t+1},\dx_{t})\stn=\smallo(1).
$$ 

%

To see that the $\smallo(1)$ term in the above expression does not
depend on $\bx_{t+1}$, note that in \eqref{eq:somejoint},
the dominating measure does not depends on $\bx_{t+1}$, and the density 
with respect to this dominating measure is bounded uniformly with respect to $\bx_{t+1}$,
and therefore the results follows from the computations in the proof
of \citet[][Theorem 1]{SQMC}. This shows \eqref{eq:Smooth} for $t\geq
1$. For $t=0$ replace $\Q_{t-1,h}\otimes m_{t,h}$ by $m_{0,h}$ in the above argument.

Let us now prove the second part of the theorem. As a preliminary
result to establish \eqref{eq:Smooth2} we show that, for all $t\geq
0$,
\begin{align}\label{eq:ind}
\|\Qh_{t+1}^N\otimes
\mathcal{M}_{t+1,\Qh_{t}^N}-\Q_{t+1}\otimes\mathcal{M}_{t+1,\Q_t}\stn=\smallo(1).
\end{align}

Let $B_t$ and $B_{t+1}$ be two sets in $\mathcal{B}_{\ui^d}$ and note
$B_{t:t+1}=B_t\times B_{t+1}$ to simplify the notations. Then,
\begin{align*}
&\left|\Qh_{t+1}^N\otimes
  \mathcal{M}_{t+1,\Qh_{t}^N}(B_{t:t+1})-\Q_{t+1}\otimes\mathcal{M}_{t+1,\Q_t}(B_{t:t+1})\right|\\ &=\left|\int_{B_{t+1}}\lambda_d\left(F_{\mathcal{M}_{t+1,\Qh_{t}^N}}(\bx_{t+1},B_t)\right)
  \Qh_{t+1}^N(\dx_{t+1})-\lambda_d\left(F_{\mathcal{M}_{t+1,\Q_{t}}}(\bx_{t+1},B_t)\right)\Q_{t+1}(\dx_{t+1})\right|\\ &\leq
  \left|\int_{B_{t+1}}\lambda_d\left(F_{\mathcal{M}_{t+1,\Q_{t}}}(\bx_{t+1},B_t)\right)
  \left(\Qh_{t+1}^N-
  \Q_{t+1}\right)(\dx_{t+1})\right|\\ &+\left|\int_{B_{t+1}}\Qh_{t+1}^N(\dx_{t+1})
  \left[\lambda_d\left(F_{\mathcal{M}_{t+1,\Qh_{t}^N}}(\bx_{t+1},B_t)\right)-\lambda_d\left(F_{\mathcal{M}_{t+1,\Q_{t}}}(\bx_{t+1},B_t)\right)\right]\right|.
\end{align*}
By assumption, $F_{\mathcal{M}_{t+1,\Q_{t}}}(\bx_{t+1},\bx_{t})$ is
H\"{o}lder continuous. Since $\|\Qh_{t+1}^N-\Q_{t+1}\stn=\smallo(1)$
by Theorem \ref{thm:consistency2}, Lemma \ref{lemma:Holder_B}
therefore implies
$$
\sup_{B_{t:t+1}\in\mathcal{B}^2_{\ui^d}}\left|\int_{B_{t+1}}\lambda_d\left(F_{\mathcal{M}_{t+1,\Q_{t}}}(\bx_{t+1},B_t)\right)
\left(\Qh_{t+1}^N- \Q_{t+1}\right)(\dx_{t+1})\right|=\smallo(1).
$$ In addition,
\begin{align*}
&\left|\int_{B_{t+1}}\Qh_{t+1}^N(\dx_{t+1})
  \left[\lambda_d\left(F_{\mathcal{M}_{t+1,\Qh_{t}^N}}(\bx_{t+1},B_t)\right)-\lambda_d\left(F_{\mathcal{M}_{t+1,\Q_{t}}}(\bx_{t+1},B_t)\right)\right]\right|\\ &\leq
  \int_{B_{t+1}}\Qh_{t+1}^N(\dx_{t+1}) \sup_{B_t\in
    \mathcal{B}_{\ui^d}}\left|\lambda_d\left(F_{\mathcal{M}_{t+1,\Qh_{t}^N}}(\bx_{t+1},B_t)\right)-\lambda_d\left(F_{\mathcal{M}_{t+1,\Q_{t}}}(\bx_{t+1},B_t)\right)\right|\\ &\leq
  \int_{B_{t+1}}\Qh_{t+1}^N(\dx_{t+1})
  \sup_{\bx_{t+1}\in\ui^d}\|\mathcal{M}_{t+1,\Qh_{t}^N}(\bx_{t+1},\dx_t)-\mathcal{M}_{t+1,\Q_{t}}(\bx_{t+1},\dx_t)\stn\\ &=\smallo(1)
\end{align*}
using \eqref{eq:Smooth}. This complete the proof of (\ref{eq:ind}).

We are now ready to prove the second statement of the theorem. Note
that \eqref{eq:Smooth2} is true for $t=1$ by \eqref{eq:ind}. Let $t>1$
and $B_{0:t}\in\mathcal{B}^{t+1}_{\ui^{d}}$. Then,
\begin{align*}
&\left|\int_{B_{0:t}}\left(\Qt_{t}^{N}-\Qt_t\right)(\dx_{0:t})\right|=\left|\int_{B_{0:t}}\left(\Qh_{t}^N\otimes
  \mathcal{M}_{t,\Qh^N_{t-1}}(\dx_{t-1:t})\prod_{s=1}^{t-1}
  \mathcal{M}_{s,\Qh_{s-1}^N}(\bx_{s},\dx_{s-1})\right.\right.\\ &-\left.\left. \Q_t\otimes
  \mathcal{M}_{t,\Q_{t-1}}(\dx_{t-1:t})\prod_{s=1}^{t-1}
  \mathcal{M}_{s,\Q_{s-1}}(\bx_{s},\dx_{s-1})\right)\right|\\ &\leq
  \left|\int_{B_{t-1:t}}\left[\int_{B_{0:t-2}}
    \prod_{s=1}^{t-1}\mathcal{M}_{s,\Q_{s-1}}(\bx_{s},\dx_{s-1})\right]\left(
  \Qh_{t}^N\otimes \mathcal{M}_{t,\Qh_{t-1}^N}-\Q_{t}\otimes
  \mathcal{M}_{t,\Q_{t-1}}\right)(\dx_{t-1:t})\right|\\ &+\left|\int_{B_{t-1:t}}\Qh_{t}^N\otimes
  \mathcal{M}_{t,\Qh_{t-1}^N}(\dx_{t-1:t})\left(\int_{B_{0:t-2}}\prod_{s=1}
  ^{t-1}\mathcal{M}_{s,\Qh^N_{s-1}}(\bx_{s},\dx_{s-1})-\int_{B_{0:t-2}}\prod_{s=1}^{t-1}
  \mathcal{M}_{s,\Q_{s-1}}(\bx_{s},\dx_{s-1})\right) \right|.
\end{align*}

The first term after the inequality sign can be rewritten as
$$
\left|\int_{B_{t-1:t}}\lambda_{(t-1)d}\left(F_{\otimes_{s=1}^{t-1}\mathcal{M}_{s,\Q_{s-1}}}(\bx_{t-1},B_{0:t-2})\right)
\left( \Qh_{t}^N\otimes \mathcal{M}_{t,\Qh_{t-1}^N}-\Q_{t}\otimes
\mathcal{M}_{t,\Q_{t-1}}\right)(\dx_{t-1:t})\right|.
$$ The supremum of this quantity over
$B_{0:t}\in\mathcal{B}^{t+1}_{\ui^{d}}$ is $\smallo(1)$ using
\eqref{eq:ind}, the fact that
$F_{\otimes_{s=1}^{t-1}\mathcal{M}_{s,\Q_{s-1}}}$ is H\"{o}lder
continuous (because $F_{\mathcal{M}_{s,\Q_{s-1}}}$ is H\"{o}lder
continuous for all $s$) and Lemma \ref{lemma:Holder_B}.

To control the second term we first prove by induction that, for any
$t>1$,
\begin{align}\label{eq:bc}
\sup_{B_{0:t-2}\in\mathcal{B}^{t-1}_{\ui^d}}
\left|\int_{B_{0:t-2}}\prod_{s=1}
^{t-1}\mathcal{M}_{s,\Qh^N_{s-1}}(\bx_{s},\dx_{s-1})-\int_{B_{0:t-2}}\prod_{s=1}^{t-1}
\mathcal{M}_{s,\Q_{s-1}}(\bx_{s},\dx_{s-1})\right|=\smallo(1)
\end{align}
 uniformly on $\bx_{t-1}$. By \eqref{eq:Smooth} this result is true
 for $t=2$. Assume that \eqref{eq:bc} holds for $t>2$. Then
\begin{align*}
&\left|\int_{B_{0:t-1}}\prod_{s=1}^{t}\mathcal{M}_{s,\Qh^N_{s-1}}(\bx_{s},\dx_{s-1})-\int_{B_{0:t-1}}
  \prod_{s=1}^{t}
  \mathcal{M}_{s,\Q_{s-1}}(\bx_{s},\dx_{s-1})\right|\\ &=\left|\int_{B_{0:t-1}}\left[\mathcal{M}_{t,\Qh^N_{t-1}}(\bx_{t},\dx_{t-1})
    \prod_{s=1}^{t-1}\mathcal{M}_{s,\Qh^N_{s-1}}(\bx_{s},\dx_{s-1})\right.\right.\\ &-\left.\left. \mathcal{M}_{t,\Q_{t-1}}(\bx_{t},\dx_{t-1})
    \prod_{s=1}^{t-1}
    \mathcal{M}_{s,\Q_{s-1}}(\bx_{s},\dx_{s-1})\right]\right|\\ &\leq
  \left|\int_{B_{t-1}}\mathcal{M}_{t,\Qh^N_{t-1}}(\bx_{t},\dx_{t-1})
  \int_{B_{0:t-2}}\left(\prod_{s=1}^{t-1}
  \mathcal{M}_{s,\Qh^N_{s-1}}(\bx_{s},\dx_{s-1})- \prod_{s=1}^{t-1}
  \mathcal{M}_{s,\Q_{s-1}}(\bx_{s},\dx_{s-1})\right)\right|\\ &+\left|\int_{B_{t-1}}
  \lambda_{(t-1)d}\left(F_{\otimes_{s=1}^{t-1}\mathcal{M}_{s,\Q_{s-1}}}(\bx_{t-1},B_{0:t-2})\right)
  \left(\mathcal{M}_{t,\Qh^N_{t-1}}(\bx_{t},\dx_{t-1})-\mathcal{M}_{t,\Q_{t-1}}(\bx_{t},\dx_{t-1})\right)\right|
\end{align*}
where we saw above that second term on the right side of the
inequality sign is $\smallo(1)$ uniformly on $\bx_t$ while the first
term is bounded by
\begin{align*}
&\int_{\ui^d}\mathcal{M}_{t,\Qh^N_{t-1}}(\bx_{t},\dx_{t-1})\\ &\times
  \sup_{B_{0:t-2}
    \in\mathcal{B}^{t-1}_{\ui^d}}\left|\int_{B_{0:t-2}}\left(\prod_{s=1}^{t-1}
  \mathcal{M}_{s,\Qh^N_{s-1}}(\bx_{s},\dx_{s-1})- \prod_{s=1}^{t-1}
  \mathcal{M}_{s,\Q_{s-1}}(\bx_{s},\dx_{s-1})\right)\right|
\end{align*}
where, by the inductive hypothesis, the second factor is $\smallo(1)$
uniformly on $\bx_{t-1}\in\ui^d$. This shows that \eqref{eq:bc} is
true at time $t+1$ and therefore the proof of the theorem is complete.

\subsection{Generalization of \citet{Hlawka1972}: Proof of Theorem \ref{thm:GenHM}}\label{p-thm:GenHM}

The proof of this result is an adaptation of the proof of
\citet[][``Satz 2'']{Hlawka1972}.

In what follows, we use the shorthand $\opA(B)=\Sop(\bu^{1:N})(B)=N^{-1}\sum_{n=1}^N\ind_B(u^n)$ for
any set $B\subset\ui^{d}$. One has
\[
\|\Sop(\bx^{1:N})-\pi\stn=\sup_{B\in\mathcal{B}_{[0,1)^{d}}}\left|\opA\left(F_{\pi}(B)\right)-\lambda_{d}\left(F_{\pi}(B)\right)\right|.
\]

Let $\beta=\lceil \kappa^{-1}\rceil$,
$\tilde{d}=\sum_{i=0}^{d-1}\beta^i$, $L$ an arbitrary integer, and $\mathcal{P}$ be the partition
of $[0,1)^{d}$ in $L^{\tilde{d}}$ congruent hyperrectangles $W$ of
  size $L^{-\beta^{d-1}}\times L^{-\beta^{d-2}} \times...\times L^{-1}$. Let
  $B\in\mathcal{B}_{[0,1)^{d}}$, $\mathcal{U}_{1}$ the set of the
    elements of $\mathcal{P}$ that are strictly in $F_{\pi}(B)$,
    $\mathcal{U}_{2}$ the set of elements $W\in\mathcal{P}$ such that
    $W\cap\partial(F_{\pi}(B))\neq\emptyset$, $U_{1}=\cup\text{
    }\mathcal{U}_{1}$, $U_{2}=\cup\text{ }\mathcal{U}_{2}$, and
    $U_{1}'=F_{\pi}(B)\setminus U_{1}$ so that (\citealp[``Satz
      2'']{Hlawka1972} or \citealp[Theorem 4]{SQMC})
\begin{align*}
\left|\opA\left(F_{\pi}(B)\right)-\lambda_{d}\left(F_{\pi}(B)\right)\right|&\leq|\opA(U_{1})-\lambda_{d}(U_{1})|+\#\mathcal{U}_{2}\left\{
D(\bu^{1:N})+L^{-\tilde{d}}\right\}
\end{align*}
where, under the assumption of the theorem,
$|\opA(U_{1})-\lambda_{d}(U_{1})|\leq L^{\tilde{d}-1}D(\bu^{1:N})$
\citep[see][]{Hlawka1972}.

To bound $\#\mathcal{U}_2$, we first construct a partition
$\mathcal{P}'$ of $[0,1)^{d}$ into hyperrectangles $W'$ of size $
L'{}^{-\beta^{d-1}}\times...\times L'{}^{-1}$ 
such that, for all points $\bx$ and $\bx'$ in $W'$, we have
\begin{equation}
|F_{i}(x_{1:i-1},x_i)-F_{i}(x'_{1:i-1},x'_i)|\leq
L^{-\beta^{d-i}},\quad i=1,...,d\label{eq:thmHM:cond1}
\end{equation}
where $F_{i}(x_{1:i-1},x_i)$ denotes the $i$-th component of
$F_{\pi}(\bx)$ (with $F_{i}(x_{1:i-1},x_i)=F_1(x_1)$ when $i=1$).  To
that effect, let $i\in 2:d$ and note that
\begin{multline*}
|F_{i}(x_{1:i-1},x_i)-F_{i}(x'_{1:i-1},x'_i)|\leq
|F_{i}(x_{1:i-1},x_i)-F_{i}(x_{1:i-1},x'_i)|\\ +|F_{i}(x_{1:i-1},x'_i)-F_{i}(x'_{1:i-1},x'_i)|.
\end{multline*}
By Assumption \ref{H:thmHM:3}, the probability measure
$\pi_i(x_{1:i-1},\dd x_i)$ admits a density $p_i(x_i|x_{1:i-1})$ with
respect to the Lebesgue measure such that
$\|p_i(\cdot|\cdot)\|_{\infty}<+\infty$. Therefore, the first term
after the inequality sign is bounded by $\|p_i\|_{\infty}
L'{}^{-\beta^{d-i}}$.  For the second term, the H\"{o}lder property of
$F_{\pi}$ implies that
\begin{align*}
|F_{i}(x_{1:i-1},x'_i)-F_{i}(x'_{1:i-1},x'_i)|
&\leq C_{\pi}(i-1)^{\kappa/2}(L'{}^{-\beta^{d+1-i}})^{\kappa} \\
& \leq C_{\pi}(i-1)^{\kappa/2}(L'{}^{-\beta^{d+1-i}})^{1/\beta}
=C_{\pi}(i-1)^{\kappa/2} L'{}^{-\beta^{d-i}}
\end{align*}
with  $C_{\pi}$ the H\"{o}lder constant of $F_{\pi}$. For $i=1$, we
simply have
$$ |F_{1}(x_1)-F_{1}(x_1')|\leq \|p_1\|_{\infty}L'{}^{-\beta^{d-1}}.
$$ 
Condition \eqref{eq:thmHM:cond1} is therefore verified for $L'$ the
smallest integer such that $L'\geq \tilde{C} L$, for some $\tilde{C}>0$.

Remark now that $\partial(F_{\pi}(B))=F_{\pi}(\partial (B))$ since $F$
is a continuous function. Let $R\in\partial B$ be a
$(d-1)$-dimensional face of $B$ and $\mathcal{R}$ be the set of
hyper-rectangles $W'\in\mathcal{P}'$ such that $R\cap W'\neq\emptyset$.
Note that $\#\mathcal{R}\leq L'{}^{\tilde{d}-1}\leq (\lfloor \tilde{C}
L\rfloor +1)^{\tilde{d}-1}$. For each $W'\in\mathcal{R}$, take a point
$\mathbf{r}^{W'}\in R\cap W'$ and define
\[
\tilde{\mathbf{r}}^{W'}=F_{\pi}(\mathbf{r}^{W'})\in F_{\pi}(R).
\]
Let $\tilde{\mathcal{R}}$ be the collection of hyper-rectangles
$\tilde{W}$ of size $2L^{-\beta^{d-1}}\times...\times2L^{-1}$ (assuming $L$ is even) and
having point $\tilde{\mathbf{r}}^{W'}$, $W'\in\mathcal{R}$, as a middle
point.

For an arbitrary $\bu\in F_{\pi}(R)$, let $\bx=F^{-1}_{\pi}(\bu)\in
R$. 
Hence, $\bx$ is in one hyperrectangle $W'\in\mathcal{R}$ so that
using \eqref{eq:thmHM:cond1}
\[
|u_i-\tilde{r}_{i}^{W'}|=|F_{i}(x_{1:i-1},x_i)-F_{i}(r_{1:i-1}^{W'},r_{i}^{W'})|\leq
L^{-\beta^{d-i}},\quad i=1,\dots,d.
\]
This shows that $\bu$ belongs to the hyperrectangle
$\tilde{W}\in\tilde{\mathcal{R}}$ with centre
$\tilde{\mathbf{r}}^{W'}$ so that $F_{\pi}(R)$ is covered by at most
$\# \tilde{\mathcal{R}}=\# \mathcal{R}\leq (\lfloor
\tilde{C}L\rfloor+1)^{\tilde{d}-1}$ hyperrectangles
$\tilde{W}\in\tilde{\mathcal{R}}$. To go back to the initial partition
of $[0,1)^{d}$ with hyperrectangles in $\mathcal{P}$, remark that
  every hyperrectangle in $\tilde{\mathcal{R}}$ is covered by at most
  $c_1$ hyperrectangles in $\mathcal{P}$ for a constant
  $c_1$. Finally, since the set $\partial B$ is made of the union of
  $2d$ $(d-1)$-dimensional faces of $B$, we have $\# \mathcal{U}_2
  \leq c_2 L^{\tilde{d}-1}$ for a constant $c_2$.

Then, we may conclude the proof as follows
\begin{align*}
\|\Sop(\bx^{1:N})-\pi\stn &\leq
L^{\tilde{d}-1}D(\bu^{1:N})+c_2L^{\tilde{d}-1}
\left(D(\bu^{1:N})+L^{-\tilde{d}}\right)
\end{align*}
where the optimal value of $L$ is such that, for some $c_3>0$, 
$$ \|\Sop(\bx^{1:N})-\pi\stn \leq c_3 D(\bu^{1:N})^{1/\tilde{d}}.
$$

\subsection{Consistency of forward smoothing: Proof of Proposition \ref{prop:constforward}}
\label{p-thm:consforward}
The proof amounts to a simple adaptation of Theorem \ref{thm:consistency2}: by
replacing Assumption 4 by Assumption 4' above, one obtains that 
$ \|\Sop(\tilde{P}^N_{t,h^{t}})-\Qt_{t-1,h^t}\otimes
m_{t,h}\stn\cvz$ as $N\rightarrow+\infty$, where
$\tilde{P}^N_{t,h^{t}}=\left(h^{t}(\hat{\bz}^{1:N}_{t-1}),\bx_t^{1:N}\right)$
 $\Qt_{t-1,h^t}$ is the image by $h^t$ of $\Qt_{t-1}$,
 and $m_{t,h}$ is defined as in Theorem \ref{thm:consistency2}. 
Therefore, by Corollary
\ref{cor:Hilbert2},
\begin{align}\label{eq:forward2}
  \|\Sop(\bz_t^{1:N})-\Qt_{t-1}\otimes m_{t}\stn\cvz,\quad\mbox{as
  }N\rightarrow+\infty.
\end{align}
In addition, since the \RN derivative
$$  \frac{\Qt_{t}}{\Qt_{t-1}\otimes   m_{t}}\left(\dd(\bx_{0:t-1},\bx_{t})\right)
\propto G_t(\bx_{t-1},\bx_{t-1}),
$$ is continuous and bounded, Theorem 1 of \citet{SQMC}, together with
\eqref{eq:forward2}, implies \eqref{eq:forward}.

\subsection{$L_2$-convergence: Proof of Theorem \ref{thm:L2}}\label{p-thm:L2}

To prove the result, let $\varphi$ be as in the statement of the
theorem and let us first prove the $L_1$-convergence.

We have
$$ \E
\left|\Sop(\tilde{\bx}_{0:T}^{1:N})(\varphi)-\Qt_T(\varphi)\right|\leq
\E
\left|\Sop(\tilde{\bx}_{0:T}^{1:N})(\varphi)-\Qt^N_T(\varphi)\right|+\E
\left|\Qt^N_T(\varphi)-\Qt_T(\varphi)\right|.
$$ By portmanteau lemma \citep[][Lemma 2.2, p.6]{VanderVaart2007}, 
convergence in the sense of the extreme metric is stronger than 
weak convergence. Hence, the second term above goes to 0 as
$N\rightarrow +\infty$ by Theorem \ref{thm:BS} and by the dominated
convergence theorem. For the first term, as each
$\tilde{\bu}^n\sim\Unif(\ui^{T+1})$, we have, by the inverse
Rosenblatt interpretation of the backward pass of SQMC,
$$
\E\Big[\Sop(\tilde{\bx}_{0:T}^{1:N})(\varphi)\big|\mathcal{F}_T\Big]=\E\Big[\Sop(\tilde{h}_{0:T}^{1:N})(\varphi\circ
  H_T)\big|\mathcal{F}_T\Big]=\Qt^N_{T,h_T}(\varphi\circ
H_T)=\Qt^N_{T}(\varphi)
$$ with $\mathcal{F}_T^N$ the $\sigma$-algebra generated by the forward
step (Algorithm \ref{alg:SQMC_B}). Therefore,
\begin{align}\label{eq:L1}
\E
\left[\big|\Sop(\tilde{\bx}_{0:T}^{1:N})(\varphi)-\Qt^N_T(\varphi)\big|\,\big|\mathcal{F}_T\right]\leq
\var\Big(\Sop(\tilde{\bx}_{0:T}^{1:N})(\varphi)\big|\mathcal{F}_T\Big)^{1/2}
\end{align}
where, using Assumption \ref{H:u2} and the fact that
$\tilde{\bx}^n_{0:T}=H_T\circ F^{-1}_{\Qt^N_{T,h_T}}(\tilde{\bu}^n)$,
\begin{align}\label{eq:L2}
\var\Big(\Sop(\tilde{\bx}_{0:T}^{1:N})(\varphi)|\mathcal{F}^N_{T})\Big)\leq
Cr(N)\sigma_{\varphi,N}^2
\end{align}
with $\sigma_{\varphi,N}^2\leq \Qt_T^{N}(\varphi^2)$ and with $C$ and $r(N)$ as in the statement of the theorem. Let
$\epsilon>0$. Then, by Assumption \ref{H:u1} and looking at the proof
of Theorem \ref{thm:BS}, we have for $N$ large enough and almost
surely, $\Qt_T^{N}(\varphi^2)\leq \Qt_T(\varphi^2)+ \epsilon$ so that,
for $N$ large enough,
\begin{align}\label{eq:L3}
\E
\left|\Sop(\tilde{\bx}_{0:T}^{1:N})(\varphi)-\Qt^N_T(\varphi)\right|\leq
\sqrt{C r(N)\big(\Qt_T(\varphi^2)+ \epsilon\big)}
\end{align}
showing the $L_1$-convergence. To prove the $L_2$-convergence, remark
that
$$
\E\left[\Sop(\tilde{\bx}_{0:T}^{1:N})(\varphi)|\mathcal{F}^N_{T}\right]=\Qt_T^{N}(\varphi)=
\big(\Qt^N_T(\varphi)-\Qt_T(\varphi)\big)+\Qt_T(\varphi)
$$ and therefore
\begin{align*}
\var\Big(\E\Big[\Sop(\tilde{\bx}_{0:T}^{1:N})(\varphi)|\mathcal{F}^N_{T}\Big]\Big)=\var\Big(\Qt^N_T(\varphi)-\Qt_T(\varphi)\Big)\leq
\E\Big[\big(\Qt^N_T(\varphi)-\Qt_T(\varphi)\big)^2\Big]
\end{align*}
where the right-hand side converges to zero as $N\rightarrow+\infty$
by the dominated convergence theorem and by Theorem \ref{thm:BS}. On
conclude the prove using \eqref{eq:L1}-\eqref{eq:L3} and the fact that
\begin{align*}
\var\Big(\Sop(\tilde{\bx}_{0:T}^{1:N})(\varphi)\Big)&=\var\Big(\E\Big[\Sop(\tilde{\bx}_{0:T}^{1:N})(\varphi)|\mathcal{F}^N_{T}\Big]\Big)+\E\Big[\var\Big(\Sop(\tilde{\bx}_{0:T}^{1:N})(\varphi)|\mathcal{F}^N_{T}\Big)\Big].
\end{align*}

\subsection{Consistency of the Backward step: Proof of Theorem \ref{thm:BS_star} and proof of Corollary \ref{cor:BS_star}}

\subsubsection{Preliminary computations}

To prove Theorem \ref{thm:BS_star} we need the following two lemmas:

\begin{lemma}\label{lemma:Hilbert}
Let $m\in\mathbb{N}$, $I=[0,\frac{k+1}{2^{dm}}]$, $k\in
\{0,1,...,2^{dm}-2\}$ and $B=H(I)$. Then, $B=\cup_{i=1}^p B_i$ for
some closed hyperrectangles $B_i\subseteq[0,1]^d$ and where $p\leq
2^d(m+1)$.
\end{lemma}

\begin{proof}
To prove the Lemma, let $0\leq m_1\leq m$ be the smallest integer
$\tilde{m}$ such that $I^d_{\tilde{m}}(0)\subseteq I$ and $i^*_{m_1}$
be the number of intervals in $\mathcal{I}^d_{m_1}$ included in
$I$. Note that $i^*_{m_1}<2^d$. Indeed, if $i^*_{m_1}\geq 2^d$ then,
by the nesting property of the Hilbert curve, $$
I^d_{m_1-1}(0)\subseteq \bigcup_{k=0}^{2^d-1} I^d_{m_1}(k)\subseteq
\bigcup_{k=0}^{i^*_{m_1}-1} I^d_{m_1}(k)\subseteq I
$$ which is in contradiction with the definition of
$i^*_{m_1}$. Define $I_2=I\setminus \cup \mathcal{I}^I_{m_1}$ and
$i^*_{m_2}$ the number of intervals in $\mathcal{I}^d_{m_2}$ included
in $I_2$. For the same reason as above $i^*_{n_2}<2^d$. More
generally, for any $m_1\leq m_k\leq m$, $i^*_{m_k}\leq 2^d$ meaning
that the set $B$ is made of at most $ \sum_{k=m_1}^m i^*_{m_k}\leq
2^d(m+1) $ hypercubes (of side varying between $2^{-m}$ and
$2^{-m_1}$).
\end{proof}


\begin{lemma}\label{lemma:genth34}
 Let $(\pi^N)_{N\geq 1}$ be a sequence of probability measures on
  $[0,1)^{(k+1)d}$ such that $\|\pi^N-\pi\stn\cvz$ where
  $\pi(\dx)=\pi(\bx)\lambda_{(k+1)d}(\dx)$ is a probability measure on
  $\pi^{(k+1)d}$ that admits a bounded density $\pi(\bx)$. Let
  $\pi_{h_k}$ be the image by $h_k$ of $\pi$. Then,
\[
\|\pi_{h_k}^N-\pi_{h_k}\stn\rightarrow0,\quad\mbox{as
}N\rightarrow+\infty.
\]

\end{lemma}

The proof of this last result is omitted since it follows from the
properties of Cartesian products and from straightforward
modifications of the proof of \citet[][Theorem 3]{SQMC}.

\subsubsection{Proof of the Theorem \ref{thm:BS_star}}\label{p-thm:BS_star}

To prove the theorem first note that
$$ \|\widetilde{\mathsf{Q}}_{T,h_T}^N-\Qt_{T,h_T}\stn=\smallo(1).
$$ Indeed, by assumption,
$\|\widehat{\mathsf{Q}}_{T,h}^N-\Qh_{T,h}^N\stn=\smallo(1)$ and, by
Theorem \ref{thm:consistency2} and \citet[][Theorem 3]{SQMC}, $\|\Qh_{T,h}^N-\Q_{T,h}\stn=\smallo(1)$ since $\Q_T$
admits a bounded density (Assumption \ref{H:thmPF1_B:4} of Theorem
\ref{thm:consistency2}). Hence,
$\|\widehat{\mathsf{Q}}_{T,h}^N-\Q_{T,h}\stn=\smallo(1)$
and thus, by Theorem \ref{thm:Hilbert2},
$\|\widehat{\mathsf{Q}}_{T}^N-\Q_{T}\stn=\smallo(1)$, with $\widehat{\mathsf{Q}}_{T}^N$  the image by $H$ of
$\widehat{\mathsf{Q}}_{T,h}^N$. In addition, using the same argument, and using the fact that,
for all $t\in 1:T$, $\tilde{G}_t$ is bounded (Assumption H1 of Theorem
\ref{thm:BS}), we have, by Theorem \ref{thm:BS} (first part),
$$
\sup_{\bx_t\in\setX}\|K^N_{t}(\bx_t,\dx_{t-1})-\mathcal{M}_{t,\Q_{t-1}}(\bx_t,\dx_{t-1})\stn=\smallo(1),\quad t\in 1:T
$$ 
with $K^N_{t}(\bx_t,\dx_{t-1})$  the image by $H$ of the
probability measure $K_{t,h}(H(\bx_t),\dd h_{t-1})$. Consequently, by the
second part of Theorem \ref{thm:BS}, $
\|\widetilde{\mathsf{Q}}_{T}^N-\Qt_T\stn=\smallo(1) $ where
$\widetilde{\mathsf{Q}}_{T}^N$ denotes the image by $H_T$ of
$\widetilde{\mathsf{Q}}_{T,h_T}^N$. Finally, under the assumptions of
the theorem, $\Qt_T$ admits a bounded density (because for all $t$,
$\tilde{G}_t$ is bounded and $\Q_t$ admits a bounded density) and
thus, by Lemma \ref{lemma:genth34}, $
\|\widetilde{\mathsf{Q}}_{T,h_T}^N-\Qt_{T,h_T}\stn=\smallo(1)$.

To prove the theorem it therefore remains to show that
\begin{align}\label{eq:thmBS_star1}
\|\Sop(\check{h}_{0:T}^{1:N})-\widetilde{\mathsf{Q}}_{T,h_T}^N\stn
=\smallo(1).
\end{align}
Indeed, this would yield
$\|\Sop(\check{h}_{0:T}^{1:N})-\Qt_{T,h_T}\stn=\smallo(1)$ and
thus, by Theorem \ref{thm:Hilbert2},\
$$
\|\Sop(\check{\bx}_{0:T}^{1:N})-\Qt_{T}\stn=\smallo(1)
$$
as required.

To prove \eqref{eq:thmBS_star1}, we assume to simplify the notations
that $F_{\mathcal{M}_{t,\Q_{t-1}}}(\bx_{t},\bx_{t-1})$ is
Lipschitz. Generalization for any H\"{o}lder exponent can be done
using similar arguments as in the proof of Theorem
\ref{thm:GenHM}.

Let $h_t^n=h(\bx_t^N)$ where $\bx_t^{1:N}$ are the particles obtained
at the end of iteration $t$ of Algorithm \ref{alg:SQMC_B}. We assume
that, for all $t\in 0:T$, the particles are sorted according to their
Hilbert index, i.e. $n< m\implies h_t^n <h_t^m$ (note that the
inequality is strict by Assumption \ref{H:thmPF1_B:1} of Theorem
\ref{thm:consistency2}). Then, using the same notation as in the proof of Theorem \ref{thm:GenHM}, one has
\[
\|\Sop(\check{h}_{0:T}^{1:N})-\widetilde{\mathsf{Q}}_{T,h_T}^N\stn=\sup_{B\in\mathcal{B}^N_{
    [0,1)^{T+1}}}\left|\opA\left(F_{\widetilde{\mathsf{Q}}_{T,h_T}^N}(B)\right)-\lambda_{T+1}\left(F_{\widetilde{\mathsf{Q}}_{T,h_T}^N}(B)\right)\right|
\]
where
$\mathcal{B}^N_{\ui^{T+1}}=\left\{[\bm{a},\bm{b}]\subset\mathcal{B}_{\ui^{T+1}},\,
b_i^N\leq h_i^N, \,i\in 0:T\right\}$.

The beginning of the proof follows the lines of Theorem
\ref{thm:GenHM}, with $\beta=d$ and $d$ replaced by $T+1$. Let
$\tilde{d}=\sum_{t=0}^Td^t$ so that, for a set
$B\in\mathcal{B}^N_{[0,1)^{T+1}}$,
\[
\left|\opA\left(F_{\widetilde{\mathsf{Q}}_{T,h_T}^N}(B)\right)-\lambda_{T+1}\left(F_{\widetilde{\mathsf{Q}}_{T,h_T}^N}(B)\right)\right|\leq
L^{\tilde{d}}D(\bu^{1:N})+\#\mathcal{U}_{2}\left\{
D(\bu^{1:N})+L^{-\tilde{d}}\right\}
\]
where $L$ and $\mathcal{U}_2$ are as in the proof of Theorem
\ref{thm:GenHM}.

Following this latter, let $\mathcal{P}'$ be the partition of the set
$[0,1)^{T+1}$ into hyperrectangles $W'$ of size $L'{}^{-d^T}\times
  L'{}^{-d^{T-1}}\times...\times L'{}^{-1}$ such that, for all
  $\bm{h}$ and $\bm{h}'$ in $W'$, we have
\begin{equation}
|F_{\widehat{\mathsf{Q}}_{T,h}^N}(h_t)-F_{\widehat{\mathsf{Q}}_{T,h}^N}(h_1')|\leq
L^{-d^T}.\label{eq:thmBS:cond1}
\end{equation}
and
\begin{equation}
\left|\tilde{F}^N_{i-1}(h_{i-1},h_{i})-\tilde{F}^N_{i-1}(h'_{i-1},h'_{i})\right|\leq
L^{-d^{T+1-i}},\quad i\in 2:(T+1)
\label{eq:thmBS:cond2}
\end{equation}
where, to simplify the notation, we write
$\tilde{F}^N_{i-1}(\tilde{h},\cdot)$ the CDF of $K^N_{T-i+2,h}(\tilde{h},\dd
h_{T-i+1})$.

Let us first look at condition \eqref{eq:thmBS:cond1}. We have
\begin{align*}
|F_{\widehat{\mathsf{Q}}_{T,h}^N}(h_1)-F_{\widehat{\mathsf{Q}}_{T,h}^N}(h_1')|
&
\leq2\|F_{\Qh_{T,h}^{N}}-F_{\widehat{\mathsf{Q}}_{T,h}^N}\|_{\infty}+2\|F_{\Qh_{T,h}^{N}}-F_{\Q_{T,h}}\|_{\infty}+|F_{\Q_{T,h}}(h_1)-F_{\Q_{T,h}}(h_1')|\\ &
\leq2r_{1}(N)+2r_{2}(N)+\left|F_{\Q_{T,h}}(h_1)-F_{\Q_{T,h}}(h_1')\right|
\end{align*}
with
$r_{1}(N)=\|F_{\widehat{\mathsf{Q}}_{T,h}^N}-F_{\Qh_{T,h}^N}\|_{\infty}$ and
$r_{2}(N)=\|\Qh_{T,h}^N-\Q_{T,h}\stn$; note $r_{1}(N)\cvz$ by the
construction of $\widehat{\mathsf{Q}}_{T,h}^N$ and under the assumptions of the theorem while
$r_{2}(N)\cvz$ by Theorem \ref{thm:consistency2} and by \citet[][Theorem 3]{SQMC}

Let $L'=2^{m}$ for an integer $m\geq0$, so that $h_i$ and $h_i'$ are
in the same interval
$I_{d^{T-i}m}^{d}(k)\in\mathcal{I}_{d^{T-i}m}^{d}$, $i\in
1:(T+1)$. Then, since $h_1$ and $h_1'$ are in the same interval
$I_{d^{T-1}m}^{d}(k)\in\mathcal{I}_{d^{T-1}m}^{d}$,
\[
\left|F_{\Q_{T,h}}(h_1)-F_{\Q_{T,h}}(h_1')\right|\leq\Q_{T,h}\left(I_{d^{T-1}m}^{d}(k)\right)=\Q_T\left(S_{d^{T-1}m}^{d}(k)\right)\leq\frac{\|p_T\|_{\infty}}{(L')^{d^T}}
\]
as $\Q_T$ admits a bounded density $p_T$. Hence \eqref{eq:thmBS:cond1}
is verified if
\[
L'\geq L \tilde{k}_N,\quad
\tilde{k}_N=\left(\frac{\|p_T\|_{\infty}}{(1-L^{d^T}r_1^*(N))}\right)^{1/d^T},\quad
r_1^*(N)=2r_{1}(N)+2r_{2}(N),
\]
which implies that we assume from now on that $L^{-d^T}\geq 2r_1^*(N)$
for $N$ large enough.

Let us now look at \eqref{eq:thmBS:cond2} for a $i>1$. To simplify the
notation in what follows, let $F^N_{i-1}(\tilde{h},\cdot)$ be the CDF of
${\mathcal{M}^{h}_{T-i+2,\Qh^N_{T-i+1,h}}}(\tilde{h},\dd h_{T-i+1})$ and
$F_{i-1}(\tilde{h},\cdot{ })$ be the CDF of
$\mathcal{M}^{h}_{T-i+2,\Q_{T-i+1}}(\tilde{h},\dd h_{T-i+1})$. Then,
\begin{align*}
 \Big|\tilde{F}^N_{i-1}&(h_{i-1},h_{i})-\tilde{F}^N_{i-1}(h_{i-1}',
  h_i')\Big|\\ & \leq
  2\|\tilde{F}^N_{i-1}-F^N_{i-1}\|_{\infty}+2\|F^N_{i-1}-F_{i-1}\|_{\infty}
  +|F_{i-1}(h_{i-1},h_{i})-F_{i-1}(h'_{i-1},h'_{i})|\\ &
  =2r_{3}(N)+2r_{4}(N)+\left|F_{i-1}(h_{i-1},h_{i})-F_{i-1}(h'_{i-1},h'_{i})\right|
\end{align*}
with $r_{3}(N)=\|\tilde{F}^N_{i-1}-F^N_{i-1}\|_{\infty}$ and
$r_{4}(N)=\|F^N_{i-1}-F_{i-1}\|_{\infty}$; note $r_{3}(N)\cvz$ by the
construction of $K^N_{T-i+2,h}$ and under the assumptions of the theorem while
$r_{4}(N)\cvz$ by Theorem \ref{thm:BS} and \citet[][Theorem 3]{SQMC}.

To control
$\left|F_{i-1}(h_{i-1},h_{i})-F_{i-1}(h'_{i-1},h'_{i})\right|$, assume
without loss of generality that $h_i\geq h'_i$ and write
$\tilde{G}_{i}^h(h_{i-1},h_i')=\tilde{G}_{T-i+2}(H(h_{i-1}), H(h_i'))$
to simplify further the notation. Then
\begin{align*}
|F_{i-1}(h_{i-1},h_{i})-F_{i-1}(h'_{i-1},h'_{i})|&\leq
|F_{i-1}(h_{i-1},h'_{i})-F_{i-1}(h'_{i-1},h'_{i})|\\ &+
\left|\int_{h_i'}^{h_i}\tilde{G}^h(h_{i-1},v)\Q_{T-i+1,h}(\dd
v)\right|.
\end{align*}
The second term is bounded by
$\|\tilde{G}_{T-i+2}\|_{\infty}\Q_{T-i+1,h}([h_i',h_i])\leq
\|\tilde{G}_{T-i+2}\|_{\infty}\Q_{T-i+1}(W)$ where
$W\in\mathcal{S}_{d^{T-i}m}^{d}$. Since $\Q_{T-i+1}$ admits a bounded
density, we have, for a constant $c>0$,
$$ \|\tilde{G}_{T-i+2}\|_{\infty}\Q_{T-i+1,h}([h_i',h_i])\leq c
L^{-d^{T+1-i}}.
$$

 To control the other term suppose first that $h_i'>L'{}^{-d^{T-i+1}}$
 and let $k$ be the largest integer such that $h'_i\geq
 kL'{}^{-d^{T-i+1}}$. Then,
\begin{equation}\label{eq:thm:BS:1}
\begin{split}
|F_{i-1}(h_{i-1},h'_{i})&-F_{i-1}(h'_{i-1},h'_{i})|\\
&= \left|\int_{0}^{h'_i}\left[\tilde{G}_i^h(h_{i-1},v)-
    \tilde{G}_i^h(h'_{i-1},v)\right]\Q_{T-i+1,h}(\dd v)\right|\\
     &\leq 
  \left|\int_{0}^{k L'{}^{-d^{T-i+1}}}\left[\tilde{G}_i^h(h_{i-1},v)-
    \tilde{G}_i^h(h'_{i-1},v)\right]\Q_{T-i+1,h}(\dd v)\right|\\
     &+
  \left|\int_{k
    L'{}^{-d^{T-i+1}}}^{h'_i}\left[\tilde{G}_i^h(h_{i-1},v)-
    \tilde{G}_i^h(t'_{i-1},v)\right]\Q_{T-i+1,h}(\dd v)\right|.
\end{split}
\end{equation}
Then, using by Lemma \ref{lemma:Hilbert}, we have for the first term:
\begin{align*}
\Big|\int_{0}^{k L'{}^{-d^{T-i+1}}}&\left[\tilde{G}_i^h(h_{i-1},v)-
    \tilde{G}_i^h(h'_{i-1},v)\right]\Q_{T-i+1,h}(\dd v)\Big|\\
 &=\left|\sum_{j=1}^{k_i}\int_{W_j}\left[\tilde{G}_{T-i+2}(H(h_{i-1}),\bx)-\tilde{G}_{T-i+2}(H(h'_{i-1}),\bx)\right]\Q_{T-i+1}(\dx)\right|\\
  &\leq
  \sum_{j=1}^{k_i}
  \Big\{\left|F^{cdf}_{\mathcal{M}_{T-i+2,\Q_{T-i+1}}}(H(h_{i-1}),\mathbf{a_j})-F^{cdf}_{\mathcal{M}_{T-i+2,\Q_{T-i+1}}}(H(h'_{i-1}),\mathbf{a_j})\right|\Big.\\ &+\Big.\left|F^{cdf}_{\mathcal{M}_{T-i+2,\Q_{T-i+1}}}(H(h_{i-1}),\mathbf{b_j})-F^{cdf}_{\mathcal{M}_{T-i+2,\Q_{T-i+1}}}(H(h'_{i-1}),\mathbf{b_j})\right|\Big\}
\end{align*}
where $W_j=[\mathbf{a}_j,\mathbf{b_j}]\subset\ui^d$ and where $k_i\leq
2^{d}(d^{T-i}m+1)$. Let $C_i$ be the Lipschitz constant of
$F^{cdf}_{\mathcal{M}_{T-i+2,\Q_{T-i+1}}}$. Then, for any
$\mathbf{c}\in\ui^d$,
\begin{align*}
\left|F^{cdf}_{\mathcal{M}_{T-i+2,\Q_{T-i+1}}}(H(h_{i-1}),\mathbf{c})-F^{cdf}_{\mathcal{M}_{T-i+2,\Q_{T-i+1}}}(H(h'_{i-1}),\mathbf{c})\right|&\leq
C_i\|H(h_{i-1})-H(h_{i-1}')\|_{\infty}\\ &\leq C_i L'{}^{-d^{T-i+1}}
\end{align*}
because $H(h_{i-1})$ and $H(h'_{i-1})$ belong to the same hypercube
$W\in\mathcal{S}_{d^{T-i}m}^d$ of side
$2^{-md^{T-i+1}}=L'{}^{-d^{T-i+1}}$.

For the second term after the inequality sign in \eqref{eq:thm:BS:1},
we have
\begin{align*}
&\left|\int_{kL'{}^{-d^{T-i+1}}}^{h'_i}\left[\tilde{G}_i^h(h_{i-1},v)-\tilde{G}_i^h(h'_{i-1},v)
    \right]\Q_{T-i+1,h}(\dd v)\right|\\ &\leq
  \mathcal{M}^{h}_{T-i+2,\Q_{T-i+1,h}}(h_{i-1},[kL'{}^{-d^{T-i+1}},h_i'])+\mathcal{M}^{h}_{T-i+2,\Q_{T-i+1,h}}(h'_{i-1},[kL'{}^{-d^{T-i+1}},h_i'])\\ &\leq
  \mathcal{M}_{T-i+2,\Q_{T-i+1}}(H(h_{i-1}),W)+\mathcal{M}_{T-i+2,\Q_{T-i+1}}(H(h'_{i-1}),W)\\ &\leq
  2\|\tilde{G}_{T-i+2}\|_{\infty}\|p_{T-i+1}\|_{\infty}2^{-md^{T-i+1}}
\end{align*}
for a $W\in\mathcal{S}_{d^{T-i}m}^{d}$ and where $p_{T-i+1}$ is the
(bounded) density of $\Q_{T-i+1}$. This last quantity is also the
bound we obtain for $h_i'<L'{}^{-d^{T-i+1}}$. Hence, these
computations shows that
$$
|\tilde{F}_{i-1}(h_{i-1},h_{i})-\tilde{F}_{i-1}(h'_{i-1},h'_{i})|\leq
c_i L'{}^{-d^{T-i+1}}\log(L')
$$ for a constant $c_i$, $i\in 2:(T+1)$.

Condition \eqref{eq:thmBS:cond2} is therefore verified when (taking $L'$ so that $\log (L')\geq 1$)
$$ \frac{L'}{\log(L')}\geq L
\max_{i\in\{2,...,T+1\}}\left(\frac{c_i}{1-L^{d^{T-i+1}}r_2^*(N))}\right)^{\frac{1}{d^{T-i+1}}}
$$ where $r_2^*(N)=2r_3(N)+2r_4(N)$.  Let $\gamma\in(0,1)$ and note
that for $N$ large enough $\log L'< L'{}^{\gamma}$. Hence, for $N$
large enough \eqref{eq:thmBS:cond1} and \eqref{eq:thmBS:cond2} are
verified for $L'$ the smallest power of 2 such that
$$ L'\geq (k_NL)^{(1-\gamma)^{-1}},\quad k_N=
\max_{i\in\{1,...,T+1\}}\left(\frac{c_i}{1-L^{d^{T-i+1}}r^*(N))}\right)^{\frac{1}{d^{T-i+1}}},\quad
c_1=\|p_T\|
$$ where $r^*(N)=r_1^*(N)+ r_2^*(N)$. Note that we assume from now on
that $L^{-d^T}\geq 2r^*(N)$.

Because the function $F_{\widetilde{\mathsf{Q}}_{T,h_T}^N}$ is
continuous on $[0,h^N_0]\times\dots\times [0,h^N_T]$,
$\partial(F_{\widetilde{\mathsf{Q}}_{T,h_T}^N}(B))=F_{\widetilde{\mathsf{Q}}_{T,h_T}^N}(\partial(B))$
and therefore we can bound $\#\mathcal{U}_{2}$ following the proof of
Theorem \ref{thm:GenHM}. Using the same notations as in the proof of
Theorem \ref{thm:GenHM}, we obtain that
$\widetilde{\mathsf{Q}}_{T,h_T}^N(\partial(B))$ is covered by at most
$$
(T+1)2^{\tilde{d}}k_{N}^{\frac{\tilde{d}-1}{1-\gamma}}L^{\frac{\tilde{d}-1}{1-\gamma}}
$$ hyperrectangles in $\tilde{\mathcal{R}}$. To go back to the initial
partition of $[0,1)^{T+1}$ with hyperrectangles $W\in\mathcal{P}$,
  remark that $L'>L$ so that every hyperrectangles in
  $\tilde{\mathcal{R}}$ is covered by at most $c^*$ hyperrectangles of
  $\mathcal{P}$ for a constant $c^*$. Hence,
\begin{align}\label{eq:thm:LD_U}
\# \mathcal{U}_2^{(1)}\leq c_N L^{\frac{\tilde{d}-1}{1-\gamma}},\quad
c_N=c^*(T+1)2^{\tilde{d}}k_{N}^{\frac{\tilde{d}-1}{1-\gamma}}.
\end{align}

We therefore have
\begin{align*}
\|\Sop(\check{h}_{0:T}^{1:N})-\widetilde{\mathsf{Q}}_{T,h_T}^N\stn
&\leq L^{\tilde{d}}D(\bu^{1:N})+c_NL^{\frac{\tilde{d}-1}{1-\gamma}}
\left(D(\bu^{1:N})+L^{-\tilde{d}}\right).
\end{align*}
Let $\gamma \in (0,\tilde{d}^{-1})$ so that
$c_d:=\tilde{d}-\frac{\tilde{d}-1}{1-\gamma}>0$. To conclude the proof
as in \citet[][Theorem 4]{SQMC}, let $\tilde{d}_1=d^T$ and
$\tilde{d}_2=\sum_{t=0}^{T-1}d^t$. Thus,
$$
\|\Sop(\check{h}_{0:T}^{1:N})-\widetilde{\mathsf{Q}}_{T,h_T}^N\stn\leq
2L^{\tilde{d}_1+\tilde{d}_2}D(\bu^{1:N})+c_NL^{-c_d}
$$ where the optimal value of $L$ is such that
$L=\bigO(D(\bu^{1:N})^{-\frac{1}{c_d+\tilde{d}_1+\tilde{d}_2}})$. Then,
provided that
$r^*(N)D(\bu^{1:N})^{-\frac{\tilde{d}_1}{c_d+\tilde{d}_1+\tilde{d}_2}}=\bigO(1)$,
$L$ verifies all the conditions above and, since $c_N=\bigO(1)$, we
have
\[
\|\Sop(\check{h}_{0:T}^{1:N})-\widetilde{\mathsf{Q}}_{T,h_T}^N\stn=\bigO\left(D(\bu^{1:N})^{\frac{1}{c_d+\tilde{d}_1+\tilde{d}_2}}\right).
\]
Otherwise, if
$r^*(N)D(\bu^{1:N})^{-\frac{\tilde{d}_1}{c_d+\tilde{d}_1+\tilde{d}_2}}\rightarrow+\infty$,
let $L=\bigO(r^*(N)^{-\frac{1}{\tilde{d}_1}})$. Then $c_{N}=\bigO(1)$
and
\begin{align*}
 L^{\tilde{d}_1+\tilde{d}_2}D(\bu^{1:N}) &
 =\bigO(r(N))^{\frac{c_d}{\tilde{d}_1}-\frac{c_d+\tilde{d}_1+\tilde{d}_2}{\tilde{d}1}}D(\bu^{1:N})
 \\ & =\bigO(r(N)^{c_d/\tilde{d}_1})
 \left(\bigO(r(N))^{-1}D(\bu^{1:N})^{\frac{\tilde{d}_1}{c_d+\tilde{d}_1+\tilde{d}_2}}\right)^{\frac{c_d+\tilde{d}_1+\tilde{d}_2}{\tilde{d}_1}}\\ &
 =\smallo\left(r(N)^{c_d/\tilde{d}_1}\right).
\end{align*}
Therefore
$\|\Sop(\check{h}_{0:T}^{1:N})-\widetilde{\mathsf{Q}}_{T,h_T}^N\stn=\smallo(1)$,
which concludes the proof.

\subsubsection{Proof of the Corollary \ref{cor:BS_star}}\label{p-cor:BS_star}

To prove the result we first construct a probability measure
$\widetilde{\mathsf{Q}}_{T,h_T}^N$ such that the point set
$\tilde{\bx}_{0:T}^{1:N}$ generated by Algorithm \ref{alg:back}
becomes, as $N$ increases, arbitrary close to the point set
$\check{\bx}_{0:T}^{1:N}$ obtained using a smooth backward step described in Theorem \ref{thm:BS_star}. Then, we show that, if $\|\Sop(\check{\bx}_{0:T}^{1:N})-\Qt_T\stn\cvz$, then$
\|\Sop(\tilde{\bx}_{0:T}^{1:N})-\Qt_T\stn\cvz$.

To this aims, assume that, for all $t\in 0:T$, the points $h_t^{1:N}$ are labelled so
that $n<m\implies h_t^n < h_t^m$. (Note that the inequality is strict
because, by Assumption \ref{H:thmPF1_B:1} of Theorem
\ref{thm:consistency2}, the points $\bx_t^{1:N}$ are distinct.)
Without loss of generality, assume that $h_t^1>0$ and let $h_t^0=0$
for all $t$.

To construct $\widetilde{\mathsf{Q}}_{T,h_T}^N$, let
$\widehat{\mathsf{Q}}_{T,h}^N$ be such that
$F_{\widehat{\mathsf{Q}}_{T,h}^N}$ is strictly increasing on
$[0,h_T^N]$ with
$F_{\widehat{\mathsf{Q}}_{T,h}^N}(h_T^n)=F_{\Qh_{T,h}^N}(h_T^n)$ for
all $n\in 1:N$ and, for $t\in 1:T$, let $K^N_{t,h}(h_t,\dd h_{t-1})$
be such, for all $h_t\in\ui$,
$F_{K^N_{t,h}}(h_t,\cdot{})$ is strictly increasing on $[0,h_{t-1}^N]$
and
$$ F_{K^N_{t,h}}(h_t,
h^n_{t-1})=F_{\mathcal{M}^{h}_{t,\Qh^N_{t-1,h}}}(h_t, h^n_{t-1}),\quad \forall n\in 1:N.
$$ 

Let $\check{h}_{0:T}^{1:N}$ be as in Theorem \ref{thm:BS_star}
(with $\widetilde{\mathsf{Q}}_{T,h_T}^N$ constructed using the 
above choice of $\widehat{\mathsf{Q}}_{T,h}^N$ and $K^N_{t,h}(h_1,\dd
h_{t-1})$). We now show by a backward induction that, for any $t\in
0:T$, $\max_{n\in
  1:N}\|\check{\bx}_t^n-\tilde{\bx}_t^n\|_{\infty}=\smallo(1)$.

To see this, note that, by the construction of
$\widehat{\mathsf{Q}}_{T,h}^N$,
$$ |\check{h}_T^n-\tilde{h}_T^n|\leq \Delta_T^N,\quad
\Delta_T^N:=\max_{n\in 1:N}|h_T^{n-1}-h_T^{n}|
$$ where, by \citet[][Lemma 2]{SQMC}, $\Delta_T^N\cvz$ as
$N\rightarrow+\infty$. Hence, using the H\"older property of the
Hilbert curve, this shows that $\max_{n\in
  1:N}\|\check{\bx}_T^n-\tilde{\bx}_T^n\|_{\infty}=\smallo(1)$.

Let $t\in 0:T-1$ and assume that $\max_{n\in
  1:N}\|\check{\bx}_{t+1}^n-\tilde{\bx}_{t+1}^n\|_{\infty}=\smallo(1)$. Let
$w_t^n=h(\bx_t^{\check{a}_t^n})$, where $\check{a}_t^n$ is the  index
selected at iteration $t$ of Algorithm \ref{alg:back} obtained by
replacing $\tilde{\bx}_{t+1}^n$ by $\check{\bx}_{t+1}^n$. Then, by the
construction of $K^N_{t,h}$, $ \max_{n\in
  1:N}|w_t^n-\check{h}_T^n|=\smallo(1)$.

We now want to show that $\max_{n\in
  1:N}|w_t^n-\tilde{h}_T^n|=\smallo(1)$. To simplify the notation, let $\tilde{m}_{t+1}(\bx_t,\bx_{t+1})=m_{t+1}(\bx_t,\bx_{t+1})G_{t_1}(\bx_t,\bx_{t+1})$. Then, using Assumption \ref{H:corr:gamma}, simple computations show that, for $m\in 1:N$,

\begin{align*}
&|\widetilde{W}^m_t(\tilde{\bx}_{t+1}^n)-\widetilde{W}^m_t(\check{\bx}_{t+1}^n)|
\leq \frac{\left|W_t^m\tilde{m}_{t+1}(\bx_t^m,\tilde{\bx}_{t+1}^n)-W_t^m\tilde{m}_{t+1}(\bx_t^m,\check{\bx}_{t+1}^n)\right|}{\sum_{k=1}^NW_t^k\tilde{m}_{t+1}(\bx_t^k,\tilde{\bx}_{t+1}^n)}\\
&+W_t^m\tilde{m}_{t+1}(\bx_t^m,\check{\bx}_{t+1}^n)\frac{\left|\sum_{k=1}^NW_t^k\tilde{m}_{t+1}(\bx_t^k,\check{\bx}_{t+1}^n)-\sum_{k=1}^NW_t^k\tilde{m}_{t+1}(\bx_t^k,\tilde{\bx}_{t+1}^n)\right|}{\big(\sum_{k=1}^NW_t^k\tilde{m}_{t+1}(\bx_t^k,\tilde{\bx}_{t+1}^n)\big)\big(\sum_{k=1}^NW_t^k\tilde{m}_{t+1}(\bx_t^k,\check{\bx}_{t+1}^n)\big)}\\
&\leq \|G_t\|_{\infty}\frac{|\tilde{m}_{t+1}(\bx_t^m,\tilde{\bx}_{t+1}^n)-\tilde{m}_{t+1}(\bx_t^m,\check{\bx}_{t+1}^n)|}{N\underline{c}_t}\\
&+\|G_t\tilde{m}_{t+1}\|_{\infty}\frac{\sum_{k=1}^NG_t(\hat{\bx}_{t-1}^k,\bx_t^k)\left|\tilde{m}_{t+1}(\bx_t^k,\check{\bx}_{t+1}^n)-\tilde{m}_{t+1}(\bx_t^k,\tilde{\bx}_{t+1}^n)\right|}{(N\underline{c}_t)^2}.
\end{align*}
Let 
$$
\omega_{t+1}(\delta)=\sup_{
\substack{(\bx_1,\bx_2)\in\setX^2,\,(\bx'_1,\bx'_2)\in\setX^2\\
\|\bx_i-\bx'_i\|_{\infty}\leq\delta,\,i=1,2
}}|\tilde{m}_{t+1}(\bx_1,\bx_2)-\tilde{m}_{t+1}(\bx'_1,\bx'_2)|,\quad \delta>0
$$
be the modulus of continuity of $\tilde{m}_{t+1}$. Then,  
\begin{align*}
|\widetilde{W}^i_t(\tilde{\bx}_{t+1}^n)-\widetilde{W}^i_t(\check{\bx}_{t+1}^n)|&\leq \max_{n\in 1:N} \frac{w_{t+1}(|\tilde{\bx}_{t+1}^n-\check{\bx}_{t+1}^n|_{\infty})}{N} \frac{\|G_t\|_{\infty}(\underline{c}_t+\|G_t\tilde{m}_{t+1}\|_{\infty})}{\underline{c}_t^2}\\
&=:\tilde{\xi}_t^N
\end{align*}
where,  using the fact that $\tilde{m}_{t+1}$ is uniformly continuous on $\setX^2$ (Assumption \ref{H:corr:Unif}) and the inductive hypothesis, $\tilde{\xi}_t^N=\smallo(N^{-1})$. Also, we know that
$$
\min_{m\in 1:N}\inf_{\bx_{t+1}\in\setX}\widetilde{W}^m_t(\bx_{t+1})\geq \xi^N_{t}:= \frac{\underline{c}_t}{N \|G_t\|\tilde{m}_{t+1}\|_{\infty}}.
$$
Then, let $N_t$ be such that $\tilde{\xi}_t^{N_t}<\xi_t^{N_t}$ so that, for $N\geq N_t$, we either have $\tilde{h}_t^n=w_t^n$, or $\tilde{h}_t^n=w_t^{n+1}$ or $\tilde{h}_t^n=w_t^{n-1}$. Hence, $\max_{n\in 1:N}|w_t^n-\tilde{h}_t^n|=\smallo(1)$ and therefore $\max_{n\in 1:N}|\tilde{h}_t^n-\check{h}_t^n|=\smallo(1)$. Finally, by, the H\"older property of the Hilbert curve, this shows that $\max_{n\in 1:N}\|\check{\bx}_{t}^n-\tilde{\bx}_{t}^n\|_{\infty}=\smallo(1)$.

The rest of the proof follows the lines of \citet[][Lemma 2.5,
  p.15]{Niederreiter1992}. First, note that the above computations
shows that, for any $\epsilon>0$, there exists a $N_{\epsilon}$ such
that
$\|\tilde{\bx}_{0:T}^{1:N}-\check{\bx}_{0:T}^{1:N}\|_{\infty}\leq\epsilon$ for $N\geq N_{\epsilon}$. Let
$B=[\bm{a},\bm{b}]$, $B^+=[\bm{a},\bm{b}+\epsilon]\cap \ui^{T+1}$ and
$B^-=[\bm{a},\bm{b}-\epsilon]$. If $\epsilon>b_i$ for at least one
$i\in 1:(T+1)$, $B^-=\emptyset$. Then for $N\geq N_{\epsilon}$, we
have

\begin{equation}\label{eq: ineg1}
\Sop(\check{\bx}_{0:T}^{1:N})(B^-)\leq
\Sop(\tilde{\bx}_{0:T}^{1:N})(B)\leq
\Sop(\check{\bx}_{0:T}^{1:N})(B^+).
\end{equation}
By the definition of the extreme metric, we have
\begin{equation}\label{eq: ineg2}
\begin{split}
&\left|\Sop(\check{\bx}_{0:T}^{1:N})(B^+)-\Qt_T(B^+)\right|\leq
  \|\Sop(\check{\bx}_{0:T}^{1:N})-\Qt_T\stn,\\ &\left|\Sop(\check{\bx}_{0:T}^{1:N})(B^-)-\Qt_T(B^-)\right|\leq
  \|\Sop(\check{\bx}_{0:T}^{1:N})-\Qt_T\stn.
\end{split}
\end{equation}
Combining (\ref{eq: ineg1}) and (\ref{eq: ineg2}) yields:
\begin{equation}\label{eq: ineg3}
\begin{cases}
-\left(\Qt_T(B)-\Qt_T(B^-)\right)-
\|\Sop(\check{\bx}_{0:T}^{1:N})-\Qt_T\stn\leq
\Sop(\tilde{\bx}_{0:T}^{1:N})(B)-\Qt_T(B)\\ \Sop(\tilde{\bx}_{0:T}^{1:N})(B)-\Qt_T(B)\leq
\left(\Qt_T(B^+)-\Qt_T(B)\right)+\|\Sop(\check{\bx}_{0:T}^{1:N})-\Qt_T\stn.
\end{cases}
\end{equation}
Using the fact that $\Qt_T$ admits a bounded density, we have for a
constant $c>0$
\begin{equation}\label{eq: ineg4}
\begin{split}
&\Qt_T(B)-\Qt_T(B^-)\leq c\lambda_{T+1}(B\setminus B^{-})\leq c\,
  \epsilon^{T+1}\\ &\Qt_T(B^+)-\Qt_T(B)\leq
  c\lambda_{T+1}(B^+\setminus B)\leq c\, \epsilon^{T+1}.
\end{split}
\end{equation}
Therefore, combining (\ref{eq: ineg3}) and (\ref{eq: ineg4}), we
obtain, for $N\geq N_{\epsilon}$ and for all
$B\in\mathcal{B}_{\ui^{T+1}}$,
$$
- c\, \epsilon^{T+1}-\|\Sop(\check{\bx}_{0:T}^{1:N})-\Qt_T\stn\leq
\Sop(\tilde{\bx}_{0:T}^{1:N})(B)-\Qt_T(B)\leq
\|\Sop(\check{\bx}_{0:T}^{1:N})-\Qt_T\stn+c\, \epsilon^{T+1}
$$
and thus
$$
\|\Sop(\tilde{\bx}_{0:T}^{1:N})-\Qt_T \stn\leq \|\Sop(\check{\bx}_{0:T}^{1:N})-\Qt_T\stn+c\, \epsilon^{T+1}
$$
and the result follows from Theorem \ref{thm:BS_star}.

\end{document}